\documentclass[conference,compsoc]{IEEEtran}
\IEEEoverridecommandlockouts
\usepackage{cite}
\usepackage{amsmath,amssymb,amsfonts}
\usepackage{lipsum}
\usepackage{graphicx}
\usepackage{textcomp}
\usepackage{xcolor}
\def\BibTeX{{\rm B\kern-.05em{\sc i\kern-.025em b}\kern-.08em
    T\kern-.1667em\lower.7ex\hbox{E}\kern-.125emX}}

\usepackage{tikz} 
\usetikzlibrary{automata} 
\usetikzlibrary{positioning} 
\usetikzlibrary{arrows} 
\usetikzlibrary{calc} 
\usetikzlibrary{patterns} 

\newif\ifDEBUG
\DEBUGtrue

\usepackage{booktabs} 
\usepackage{amsmath}
\usepackage{algorithm}
\usepackage{algpseudocode}
\usepackage[normalem]{ulem} 
\usepackage{xspace}
\usepackage{multirow} 
\usepackage{balance} 
\usepackage{fancyvrb} 
\usepackage{caption}

\usepackage{enumitem}
\setlist[itemize]{leftmargin=*,noitemsep,topsep=0pt}
\setlist[enumerate]{leftmargin=*}

\usepackage{etoolbox}
\makeatletter
\patchcmd{\@makecaption}
	{\scshape}
	{}
	{}
	{}
\makeatletter
\patchcmd{\@makecaption}
	{\\}
	{.\ }
	{}
	{}
\makeatother
\def\tablename{Table}

\newcommand{\CSharp}{\mbox{C\#}\xspace} 

\newcommand{\REDOS}{ReDoS\xspace}

\newcommand{\ie}{\textit{i.e.,}\xspace}
\newcommand{\eg}{\textit{e.g.,}\xspace}
\newcommand{\etal}{\textit{et al.}\xspace}
\newcommand{\etals}{\textit{et al.'s}\xspace}

\usepackage{mdframed}
\mdfsetup{skipabove=0.5\topskip,skipbelow=0.5\topskip,align=center}

\usepackage{amsthm}
\newtheorem{thm}{Theorem}

\newenvironment{framedThm}
  {\begin{mdframed}\begin{thm}} {\end{thm}\end{mdframed}}

\DeclareMathSymbol{\mlq}{\mathord}{operators}{``}
\DeclareMathSymbol{\mrq}{\mathord}{operators}{`'}

\newif\ifSAVESPACE
\SAVESPACEtrue

\ifSAVESPACE
    \newcommand{\mysection}[1]{\vspace{1.5mm}\noindent\textbf{#1:}}

    \newcommand{\mysubsection}[1]{\vspace{1.5mm}\noindent\textbf{#1.}}
    \renewcommand{\mysubsection}[1]{\hspace{-0.15cm}\emph{#1.}}

    \newcommand{\myparagraph}[1]{\paragraph{#1}}
    \renewcommand{\myparagraph}[1]{\vspace{0.25em} \noindent \textbf{#1:}}
    \newcommand{\myitparagraph}[1]{\vspace{0.25em} \emph{#1:}}

    \renewcommand{\myparagraph}[1]{\mysection{#1}}
    \renewcommand{\myitparagraph}[1]{\mysubsection{#1}}
\else
    \newcommand{\mysection}[1]{\subsection{#1}}
    \newcommand{\mysubsection}[1]{\subsubsection{#1}}

\fi

\usepackage{todonotes}
\usepackage{soul}

\ifDEBUG
    \newcommand{\AH}[1]{\todo[color=cyan,inline]{AH:#1}}
    \newcommand{\JD}[1]{\todo[color=yellow,inline]{JD:#1}}
    \newcommand{\FS}[1]{\todo[color=green,inline]{FS:#1}}
    \newcommand{\DY}[1]{\todo[color=orange,inline]{DY:#1}}
    \newcommand{\ZA}[1]{\todo[color=pink,inline]{ZA:#1}}
    
\else
    \newcommand{\AH}[1]{}
    \newcommand{\JD}[1]{}
    \newcommand{\FS}[1]{}
    \newcommand{\DY}[1]{}
    \newcommand{\ZA}[1]{}

\fi

\newcommand{\HIDDEN}[1]{\todo[disable]{#1}}

\newcommand{\stacktwo}[5]{%
  \mathord{%
    \hbox{%
      \raisebox{#2}{\rlap{#1}}%
      \raisebox{#4}{\rlap{#3}}%
      \hphantom{#5}%
    }
  }
}
\newcommand{\xblnot}{\stacktwo{$\vee$}{.18ex}{$\cap$}{-.22ex}{$-$}}

\usepackage{hyperref}

\usepackage{cleveref}
\crefformat{section}{\S#2#1#3}
\crefname{figure}{Figure}{Figures}
\crefname{table}{Table}{Tables}
\crefname{theorem}{Theorem}{Theorems}
\crefname{thm}{Theorem}{Theorems}
\crefname{lemma}{Lemma}{Lemmata}
\crefname{equation}{Eqt.}{Eqts.}
\crefformat{Grammar}{Grammar #1}
\crefname{appendix}{Appendix}{Appendices}

\newtheorem{lemma}{Lemma}
\usepackage{url}

\newcommand{\Likert}[1]{\emph{``{#1}''}}

\newcommand{\myinlinequote}[1]{``\emph{#1}''}

\newcommand{\CorpusNRegexes}{537,806\xspace}
\newcommand{\AnalyzedNRegexes}{209,188\xspace}
\newcommand{\HumanSubjectsNComposition}{20\xspace}
\newcommand{\HumanSubjectsNwithTools}{9\xspace}

\begin{document}

\newcommand{\MyTitle}[1]{}
\renewcommand{\MyTitle}{Theory and Patterns for Avoiding Regex Denial of Service}
\renewcommand{\MyTitle}{Theoretically Grounded Anti-Patterns and Fix Strategies for Regex Denial of Service}
\renewcommand{\MyTitle}{Improving Developers' Understanding of Regex Denial of Service Tools through Anti-Patterns and Fix Strategies}

\title{\MyTitle}


\author{\IEEEauthorblockN{Sk Adnan Hassan\IEEEauthorrefmark{1}  \thanks{\IEEEauthorrefmark{1}Sk Adnan Hassan is currently employed at Walmart Inc.}}
\IEEEauthorblockA{\emph{Virginia Tech} \\
Blacksburg, VA, USA\\
skadnan@vt.edu}
\and
\IEEEauthorblockN{Zainab Aamir, Dongyoon Lee}
\IEEEauthorblockA{\emph{Stony Brook University} \\
Stony Brook, NY, USA\\
\{zaamir, dongyoon\}@cs.stonybrook.edu}
\and
\IEEEauthorblockN{James C. Davis}
\IEEEauthorblockA{\emph{Purdue University} \\
West Lafayette, IN, USA\\
davisjam@purdue.edu}
\and
\IEEEauthorblockN{Francisco Servant\IEEEauthorrefmark{2}  \thanks{\IEEEauthorrefmark{2}Some work performed while at Virginia Tech, U.S.A., and Universidad Rey Juan Carlos, Madrid, Spain.}}
\IEEEauthorblockA{\emph{University of M\'alaga} \\
M\'alaga, Spain \\
fservant@uma.es}
}

\maketitle
\thispagestyle{plain}



\begin{abstract}
Regular expressions are 
used for diverse purposes, including input validation and firewalls.
Unfortunately, they can also lead to a security vulnerability called \REDOS(Regular Expression Denial of Service), caused by a super-linear worst-case execution time during regex matching.
Due to the severity and prevalence of \REDOS, past work proposed automatic tools to detect and fix regexes.
Although these tools were evaluated in automatic experiments, their usability has not yet been studied; usability has not been a focus of prior work.
Our insight is that the usability of existing tools to detect and fix regexes will improve if we complement them with anti-patterns and fix strategies of vulnerable regexes.

We developed novel anti-patterns for vulnerable regexes, and a collection of fix strategies to fix them. 
We derived our anti-patterns and fix strategies from a novel theory of regex infinite ambiguity --- a necessary condition for regexes vulnerable to \REDOS.
We proved the soundness and completeness of our theory.
We evaluated the effectiveness of our anti-patterns, both in an automatic experiment and when applied manually.
Then, we evaluated how much our anti-patterns and fix strategies improve developers' understanding of the outcome of detection and fixing tools.
Our evaluation found that our anti-patterns were effective over a large dataset of regexes (N=\AnalyzedNRegexes): 100\% precision and 99\% recall, improving the state of the art 50\% precision and 87\% recall.
Our anti-patterns were also more effective than the state of the art when applied manually (N=\HumanSubjectsNComposition): 100\% developers applied them effectively vs. 50\% for the state of the art.
Finally, our anti-patterns and fix strategies increased developers' understanding using automatic tools (N=\HumanSubjectsNwithTools): from median \Likert{Very weakly} to median \Likert{Strongly} when detecting vulnerabilities, and from median \Likert{Very weakly} to median \Likert{Very strongly} when fixing them.

\end{abstract}


\HIDDEN
{

- New argumentation

Redos is a security problem that affects thousands of regexes in software, and that has created high losses for many websites, e.g., cloudflare, etc.
Given the importance of this problem, multiple approaches have been proposed to identify and fix ReDos.
However, little attention has been put into avoiding writing redos-vulnerable regexes in the first place.

points to make
	swiss cheese approach
		the more tools in our toolbox, the better
	not get infected >> testing
	existing approaches for detection
		in NFA
		in regex
			but these are not as accurate
			it may be easier to think in terms of the regex than the NFA (or at least complementary)
			NFA would require visualization, with more context switching
				humans can only keep so many things in their head at the same time
			this modality is also applicable to code review as-is (without context switching to a viz)
	regex vs. NFA similar to code vs. AST (we don't visualize the AST)
	different audiences. some people prefer the NFA, but others will prefer the regex
		can we run a small experiment for this?
	
	why not auto-detect in the background?

	why not auto-fix in the background?
		auto-fixes still require developer inspection
		these techniques are not 100\% complete

	why not auto-X in the background?
		Some people do not use tools, even if they are available
			many real-world bugs could have been found by tools, but they are not
		For that reason, we still teach developers how to do things by themselves
			example: SQL injection, and other best practices
		We have best practices that we teach for everything, even if we have tools to double-check, eg in code review

		swiss cheese approach
			some people will use tools, others will learn best practices, others will do neither
				I mean, many people still prefer vim!
				our anti-patterns require less context switching than anything based on the NFA
			that's why we need multiple tools in our toolbox
		different audiences. some people prefer the NFA, but others will prefer the regex
			can we run a small experiment for this?
		(maybe) anti-patterns can be used to improve detection tools, allowing them to provide explanations in the regex itself

	useful references
		naseem17plop - presents a taxonomy of security anti-patterns. We will use it to motivate that we are not the only ones that value security anti-patterns
		islam20secdev - set of security anti-patterns for spring framework

}

\HIDDEN
{
- Paper outline:

Infinite Ambiguity in Regular Expressions
	Introduction (1.5 pages)
		problem
		impact of the problem
		existing solutions so far
		gap in existing work
		our proposed approach(es)
		benefits of our proposed approaches
		our evaluation
		our results
		impact of our results
	Background (0.2 pages)
		ReDOS
		Super-linearity
			Existing Superlinearity anti-patterns
		Ambiguity
			Existing Ambiguity anti-patterns
	Related Work (0.5 pages)
		Studies of the impact of ReDOS
		Tools to detect ReDOS
		Tools to auto-fix vulnerability to ReDOS
	Proposed Approaches
		Infinite Ambiguity (IA) Theory (This is the _foundational_ contribution)
			Rules
			Proofs
		Infinite Ambiguity Anti-patterns (This is the _practical_ contribution)
			Anti-patterns
			How we obtained them from the theory
	Evaluation Design
		Evaluating the infinite-ambiguity rules
		- RQ-1: How effective are our infinite-ambiguity rules when applied in a large dataset of regexes?
		- RQ-2: How common is each infinite-ambiguity rule in real-world regexes?
		Evaluating the infinite-ambiguity anti-patterns
		- RQ-3: How effective are our infinite-ambiguity anti-patterns when applied in a large dataset of regexes?
		- RQ-4: How common is each infinite-ambiguity anti-pattern in real-world regexes?
		Evaluating the infinite-ambiguity anti-patterns with humans
		- RQ-5: How effective are our infinite-ambiguity anti-patterns when used by humans?
	Evaluation Dataset
		Description
		Curation
			We removed:
				- e-regexes (not supported by our theory)
				- regexes not supported by Weideman (we could not capture ground-truth for them)
	RQ-1
		Research Method
			Process: We apply the theory rules on all regexes to determine their IA status
			Ground truth: Obtained with a modified Weideman tool (it determines IA by analyzing the NFA)
			Metrics: precision, recall, F-1
		Results
			We expect to observe 100\% precision and recall
	RQ-2
		Research Method
			Process: For each regex, we label the theory rule that determined their IA status
			Metrics: Ratio of prevalence for each rule
		Results
			We expect to observe the most intuitive rules being most common, but hope to make some interesting observations, such as:
				- what domains use regexes under the less-common rules?
	RQ-3
		Research Method
			Process: We apply anti-patterns on all regexes to determine their IA status
			Studied techniques: our anti-patterns, and existing work (davis18fse)
			Ground truth: Obtained with a modified Weideman tool (it determines IA by analyzing the NFA)
			Metrics: precision, recall, F-1
		Results
			We expect our anti-patterns to provide 100\% recall, but lower precision
			We expect davis18fse to provide lower recall and precision
			We also expect to make some interesting observations, such as 
				- why are false positives in our anti-patterns happening?
				- what are some cases that our anti-patterns detect and davis18fse does not?
	RQ-4
		Research Method
			Process: For each regex, we label the anti-pattern that determined their IA status
			Studied techniques: our anti-patterns, and existing work (davis18fse)
			Metrics: Ratio of prevalence for each rule
		Results
			We expect to observe the most intuitive rules being most common, but hope to make some interesting observations, such as:
				- what domains use regexes under the less-common rules?
	RQ-5
		Research Method
			Process: experiment with users
			Main question: Are our anti-patterns more precise and complete than the existing ones (davis18fse)?
			Studied techniques: our anti-patterns, and existing work (davis18fse)
			Studied regexes: one regex that exemplifies each one of the cases that benefited from our anti-patterns (as observed in RQ-3), plus one for a case that did not benefit.
			Metrics: precision, recall
		Results	
			We expect that the benefits observed in RQ-3 also happen when humans use the tool in practice (in a lab setting)
			We also expect to capture interesting user impressions
	Discussion (0.5 pages)
		On completeness vs. usability: Theory vs. anti-patterns
			Why didn't we include the theory rules in the user study?
		The swiss cheese approach: Why not just use automated detection tools?
		On regex modalities: Why not just use EDA, IDA as anti-patterns?
	Threats to Validity (0.3 pages)
	Conclusions (0.2 pages)
	Replication
	Acknowledgements

	Questions:
		- Where do our regex-generation efforts fit (do they)?
			Tentative answer: They don't fit anywhere
}


\section{Introduction} \label{section:introduction}


Regular expressions (\emph{regexes}) are a
tool for text processing~\cite{Chapman2016RegexUsageInPythonApps,Davis2018EcosystemREDOS}.
Regexes are used across the system stack~\cite{Li2008REL,Gogte2016HARE,Chiticariu2011SystemT,efftinge2006oaw}, including in security tasks such as
input validation~\cite{Balzarotti2008Saner,Wassermann2008XSS}
and
web application firewalls~\cite{ModSecurityWAFRuleset,ClamAVRegexRules}.
Unfortunately, regexes can themselves cause a security vulnerability because of the high worst-case time complexity of backtracking-based regex engine implementations.
This algorithmic complexity vulnerability is known as 
\emph{Regular Expression Denial of Service (\REDOS)}~\cite{Crosby2003REDOS,Roichman2009ReDoS}.
For example, \REDOS caused service outages at Stack Overflow in 2016~\cite{2016StackOverflowOutage} and at Cloudflare in 2019~\cite{Cloudflare2019REDOSPostMortem}.
Researchers report hundreds of vulnerable regexes in the software supply chain~\cite{Davis2018EcosystemREDOS} and in live web services~\cite{Staicu2018REDOS,barlas2022REDOS}.

Many approaches have been proposed to address the \REDOS problem.
Our work builds on those that try to detect and fix regexes.
In this vein, some researchers characterized vulnerable regexes into anti-patterns for manual use by developers~\cite{Davis2018EcosystemREDOS}.
Others proposed tools to automatically
  detect~\cite{Berglund2014REDOSTheory,Weideman2016REDOSAmbiguity,Weideman2017StaticExpressions,Wustholz2017Rexploiter,liu21sp,Kirrage2013rxxr,Rathnayake2014rxxr2,Sugiyama2014RegexLinearityAnalysis,Sulzmann2017DerivAmbig,Shen2018ReScueGeneticRegexChecker}
  or fix~\cite{VanDerMerwe2017EvilRegexesHarmless,codykenny17,li2020flashregex,claver2021regis}
  vulnerable regexes.
All of these approaches have been evaluated solely via automatic experiments. 
Their \emph{usability} has not been studied, jeopardizing their impact in practice~\cite{johnson2013don} --- 95\% of developers reject tools when they cannot understand the results~\cite{johnson2013don}.

The goal of this paper is to improve the usability of existing \REDOS defenses.
Our insight is that 
the usability of existing tools to detect and fix regexes
will improve if we complement them with 
\emph{anti-patterns} and \emph{fix strategies} of vulnerable regexes.
We specifically aim to improve developer understanding of the outcome of the tools. 

For this goal, we developed novel anti-patterns for vulnerable regexes, and a collection of fix strategies to fix them.
We derive our anti-patterns and fix strategies from our novel theory of regex infinite ambiguity (IA).
Our theory characterizes a fundamental component of vulnerable regexes: their infinitely ambiguous (IA) region.
The IA region is what the state of the art anti-patterns characterize~\cite{Davis2018EcosystemREDOS}, what many detection tools detect, \eg~\cite{Weideman2016REDOSAmbiguity,Wustholz2017Rexploiter}, and what developers often fix in vulnerable regexes~\cite{Davis2018EcosystemREDOS}.
We refer to regexes with an IA region as IA regexes.
Our anti-patterns and fix strategies complement existing detection and fixing tools, (1) by helping developers better understand the IA region of the vulnerable regex detected by the tool;
 and (2) by providing understandable fix strategies in addition to the ones proposed by fixing tools.

Our evaluation proceeded in four phases: proving our theory, and then running three experiments.
\emph{First}, we formally proved the IA theory on which our anti-patterns are based (\cref{section:Theory} and \cref{section:Appendix-Proofs}).
\emph{Second}, since we deliberately introduced inaccuracy in our anti-patterns in favor of simplicity, we evaluated their effectiveness in an automatic experiment.
We compared our anti-patterns to the state-of-the-art ones over a large dataset of regexes (\cref{section:Experiment1}).
\emph{Third}, since low usability may lower effectiveness in manual use~\cite{johnson2013don}, we also evaluated the effectiveness of our anti-patterns when applied manually.
We compared our anti-patterns to the state-of-the-art ones
in a human-subjects experiment,
simulating a context in which
developers often prefer manual techniques~\cite{johnson2013don}, \eg when tools disrupt developer workflow, such as in regex composition or when working with simple regexes (\cref{sec:Experiment2}).
\emph{Fourth}, for more complex tasks, developers may prefer to use automatic tools.
So, we also evaluated how our anti-patterns and fix strategies complement the usage of existing automatic tools by improving their usability.
In a second human-subjects experiment (\cref{sec:Experiment3}), we measured
if our anti-patterns improve the understanding of the outcome of tools to
(a) detect
and
(b) fix
vulnerable regexes
(\cref{sec:Experiment3}).
To the best of our knowledge, this is the first study of the usability of anti-patterns or tools to detect or fix vulnerable regexes when applied by humans.


Our evaluation provided multiple findings.
\emph{First}, our underlying theory of regex infinite ambiguity was sound and complete.
\emph{Second}, our anti-patterns provided higher effectiveness (100\% precision, 99\% recall) than the state of the art anti-patterns~\cite{Davis2018EcosystemREDOS} (50\% precision, 87\% recall) over a dataset of \AnalyzedNRegexes real-world regexes \cite{Davis2019LinguaFranca}.
\emph{Third}, novice and intermediate developers (100\% of \HumanSubjectsNComposition studied) increased their effectiveness at identifying IA in regexes over 5 different regex tasks, improving from a success rate of 50\% to a rate of 100\%.
\emph{Fourth}, the \HumanSubjectsNwithTools expert developers who used our anti-patterns to complement detection tools increased their understanding of what makes a detected regex vulnerable:
  from median \Likert{Very weakly} to median \Likert{Strongly}.
Similarly, when using our fix strategies to complement fixing tools, they increased their understanding of what makes the resulting fixed regex not vulnerable:
  from median \Likert{Very weakly} to median \Likert{Very strongly}.


This paper provides the following contributions:
\begin{enumerate}
  \item A sound and complete theory of regex infinite ambiguity (\cref{section:Theory}).
  \item Derived from this theory, IA anti-patterns (\cref{section:AntiPatterns}) and IA fix strategies (\cref{sec:fix_strategies}).
  \item A quantitative evaluation of the comprehensiveness of our IA anti-patterns over the largest dataset of real-world regexes, showing that they capture IA effectively in a wide proportion of them (\cref{section:Experiment1}).
  \item The first usability evaluation of characterizations of vulnerable regexes, showing that our IA anti-patterns were usable enough for novice developers to apply effectively in the absence of tools (\cref{sec:Experiment2}).
  \item The first usability evaluation of tools to detect and fix vulnerable regexes, showing that our IA anti-patterns and IA fix strategies improved their usability by improving their understanding (\cref{sec:Experiment3}).
\end{enumerate}

Our paper provides a replication package~\cite{hassan2022} (see \cref{section:Appendix-replication}).


\section{Background} \label{section:Background}





\mysection{Regular Expressions (Regexes) and Ambiguity}
\label{section:Background:regex}

\mysubsection{Regexes}
\label{section:background:regex:regexes}
Kleene proposed regular expressions as a notation to specify a language, \ie a set of strings~\cite{kleene1951representation}.
With a finite alphabet of terminal symbols, $\Sigma$, and metacharacters, `$|$', `$\cdot$', and `$*$', 
the regular expression syntax is~\cite{Sipser2006AutomataTheory}:

\vspace{-10px}
\[ R \rightarrow \phi \;\Big{|}\; \epsilon \;\Big{|}\; \sigma \;\Big{|}\; R_1{|}R_2 \;\Big{|}\; R_1{\cdot}R_2 \;\Big{|}\; R_1* \]
\vspace{-15px}



\noindent
where $\phi$ denotes the empty language; 
$\epsilon$ is the empty string; 
the characters $\sigma \in \Sigma$ are terminal symbols;
$R_1|R_2$ alternates; 
$R_1{\cdot}R_2$ concatenates; 
and
$R\ast$ repeats.
The language function $L:R\rightarrow 2^{\Sigma^\ast}$ gives semantics:

\begin{table}[h]
\vspace{-5px}
\centering
{
\normalsize
\begin{tabular}{ll}
$L(\phi) = \phi$                & 
$L(R_1|R_2) = L(R_1) \cup L(R_2)$   \\
$L(\epsilon) = \{\epsilon\}$    &
$L(R_1{\cdot}R_2) = L(R_1){\cdot}L(R_2)$   \\
$L(\sigma) = \{\sigma\}$        &
$L(R\ast) = L(R)\ast$          
\end{tabular}
}
\vspace{-7px}
\end{table}


These semantics apply to Kleene's regexes, and extend to ``syntax sugar'' notations such as 
  character ranges \verb<[a-c]<. 
In practice, regexes may include non-regular features such as lookaround assertions, 
backreferences, 
and possessive quantifiers~\cite{Friedl2002MasteringRegexes}. 
These features are used in less than 10\% of real-world regexes~\cite{Chapman2016RegexUsageInPythonApps,Davis2019RegexGeneralizability,Davis2019LinguaFranca}.
We therefore focus on the common case of Kleene-regular regexes, denoted \emph{K-regexes}.

\mysubsection{Regex Ambiguity}
The regex language semantics allow membership to be checked with a parser.
A regex is \emph{ambiguous} if there is a string in its language that can be matched by more than one parse tree \cite{brabrand2010typed,Sulzmann2017DerivAmbig}.
For example, the regex \verb<a|a< can parse the input ``\emph{a}'' in two ways,
\ie yielding two parse trees, one using the left \verb<a< and one the right.

For K-regexes,
a regex match is equivalent to simulating an input on a corresponding non-deterministic finite automaton (NFA)~\cite{rabin1959finite}.
To simplify discussion, we will reason about regex ambiguity over an equivalent, ambiguity-preserving, $\epsilon$-free NFA~\cite{weber1991degree,Weideman2017StaticExpressions}.
From the NFA perspective, a regex is ambiguous if there is a string that can be accepted along multiple paths of this NFA. 

\mysubsection{Infinitely Ambiguous (IA) Regex}
\label{sec:definition-ia-regex}
Regexes have various degrees of ambiguity~\cite{allauzen2008general,stearns1985equivalence}:
  no ambiguity;
  finite (bounded regardless of input length);
  or
  infinite in input length.
Infinite ambiguity (IA) leads to super-linear time complexity in some parsing algorithms (\eg backtracking)~\cite{Weideman2016REDOSAmbiguity,Wustholz2017Rexploiter}. 
A regex is infinitely ambiguous if it has an infinitely ambiguous (IA) region (equivalently, an NFA section), \ie a region with the \emph{infinite-degree-of-ambiguity (IDA)} property~\cite{weber1991degree}.
Given an $\epsilon$-free finite automaton \textit{A}, necessary and sufficient conditions for \textit{A} to be infinitely ambiguous
are given by Weber \& Seidl~\cite{weber1991degree}. 

IDA can be of two types:
  (1) polynomially IDA (PDA),
  and
  (2) exponentially IDA (EDA).
\cref{fig:ida}(a) illustrates a \emph{polynomially IDA (PDA)} section in a regex's NFA.
A substring $label(\pi_i)$ can be matched in the loop $\pi_1$ at node $p$, the path $\pi_2$ from $p$ to $q$, or in the loop $\pi_3$ at $q$.
For example, consider the regex \verb<a*a*< for an input ``$aa...a$'' of length $N$.
As any two partitions of the input can be matched with the first \verb<a*< and second \verb<a*<, there are $N$ matching paths.

\cref{fig:ida}(b) illustrates an \emph{exponentially IDA (EDA)} section in a regex's NFA.
A substring $label(\pi_i)$ can be matched in either of two loops $\pi_1$ or $\pi_2$ at node $p$.
Consider the example regex \verb<(a|a)*<.
Each `$a$' of the input ``$a...a$'' can be matched by either the upper or lower loop, and thus 
the total number of matching paths becomes $2^N$. 

\begin{figure}[t] 
    
    \centering 
    \begin{minipage}[b]{0.61\columnwidth}
      \centering
      \begin{tikzpicture}[baseline=-1.80em,initial text=]
        \tikzstyle{every state}=[inner sep=1pt, minimum size=0.2cm]
        
        \node[state] (q0) {$...$};
        \node[state, right = 0.4cm of q0] (q1) {$p$}; 
        \draw [->] (q0) edge node [above] {$...$} (q1);
        \draw (q1) edge[loop above] node [left] {\tt $\pi_1$} (q1);
        \node[state, right = 0.4cm of q1] (q2) {$q$}; 
        \draw [->] (q1) edge node [above] {\tt $\pi_2$} (q2);
        \draw (q2) edge[loop below] node [left] {\tt $\pi_3$} (q2);
        \node[state, right = 0.4cm of q2] (q3) {$...$};
        \draw [->] (q2) edge node [above] {$...$} (q3);
        \node at (1,1) (nodeC) {label($\pi_1$) = label($\pi_2$) = label($\pi_3$)};
        
      \end{tikzpicture} 
      \vspace{-5px}
      \caption*{(a) PDA (or $IDA_d$)} 
    \end{minipage}
    \begin{minipage}[b]{0.37\columnwidth}
      \centering
      \begin{tikzpicture}[baseline=-1.80em,initial text=]
        \tikzstyle{every state}=[inner sep=1pt, minimum size=0.2cm]

        \node[state] (q0) {$...$};
        \node[state, right = 0.4cm of q0]     (q1) {$p$}; 
        \draw [->] (q0) edge node [above] {$...$} (q1);
        \draw (q1) edge[loop above] node [left] {\tt $\pi_1$} (q1);
        \draw (q1) edge[loop below] node [left] {\tt $\pi_2$} (q1);
        \node[state, right = 0.4cm of q1] (q2) {$...$}; 
        \draw [->] (q1) edge node [above] {$...$} (q2);
        \node at (1,1) (nodeC) {label($\pi_1$) = label($\pi_2$)};
        
      \end{tikzpicture} 
      \vspace{-10px}
      \caption*{(b) EDA} 
    \end{minipage} 
    \vspace{-15px}
    \caption{
        Illustration of Polynomial and Exponential Degree of Ambiguity (PDA, EDA) in the NFA~\cite{weber1991degree}.
        We say that $(p, \pi_2, q)$ is a transition from state $p$ to state $q$ via label($\pi_2$)~\cite{Wustholz2017Rexploiter}.
    }
    \vspace{-17px}
    \label{fig:ida}
\end{figure}

\mysection{Regex-Based Denial of Service (\REDOS)}
\label{section:Background:redos}
Regex-based Denial of Service (\REDOS)~\cite{Crosby2003REDOS} is a security vulnerability --- an algorithmic complexity attack~\cite{Crosby2003REDOS} by which 
a web service's computational resources are diverted from legitimate client interactions into an expensive regex match, degrading its quality of service.
Following Davis \etal~\cite{davis2021using},
\REDOS involves three Conditions:
\begin{enumerate}[label=(C\arabic*)]
\item a \emph{backtracking regex engine} used in evaluation, and
\item a \emph{vulnerable regex}, applied to evaluate
\item a \emph{malign input}.
\end{enumerate}

\mysubsection{C1-Backtracking Regex Engine}
Many regex engines 
(\eg versions of PHP, 
Perl, 
JavaScript, 
Java, 
Python, 
Ruby, 
and
\CSharp) 
use a backtracking search algorithm, \eg Spencer's~\cite{Spencer1994RegexEngine},
to answer regex queries~\cite{Cox2007RegExBible,Davis2019LinguaFranca}. 

\mysubsection{C2-Vulnerable Regex}
A vulnerable regex is an IA regex whose NFA has
  a \emph{prefix} region, followed by an IA region (either PDA or EDA), followed by a \emph{suffix} region~\cite{Wustholz2017Rexploiter}.
The IA region is a necessary component and the root cause of the regex's vulnerability.
The prefix must be considered to reach this IA region, and the suffix must typically lead to a mismatch in order to trigger backtracking.


\mysubsection{C3-Malign Input}
An attacker-controlled malign input triggers the super-linear behavior of a vulnerable regex by driving the backtracking engine into evaluating a polynomially or exponentially large number of possible NFA paths.
The exploration exhausts computational resources~\cite{Berglund2014REDOSTheory}. 

\vspace{0.1cm}
\mysubsection{Threat model}
We suppose the following threat model for \REDOS, aligned with the common use of regexes for input sanitization in web software~\cite{Chapman2016RegexUsageInPythonApps,Wustholz2017Rexploiter,Michael2019RegexesAreHard}.
The victim's regex engine uses a backtracking regex engine (\REDOS Condition 1),
which is common for many server-side programming languages.
The victim uses a regex (C2) to sanitize attacker-controlled input (C3).

\vspace{0.1cm}
\mysubsection{\REDOS in practice}
Davis \etal reported two high-profile examples of \REDOS affecting millions of users~\cite{davis2020impact,davis2021using}.
In~\cref{Appendix-OtherFigs} we note growing \REDOS CVEs from 2010 to present.

\section{Related Work} \label{section:RelatedWork}

\mysection{Empirical measurements of \REDOS in practice}
\label{sec:related-empirical}
Although the \REDOS attack was proposed twenty years ago by Crosby and Wallach~\cite{Crosby2003REDOS,Crosby2003AlgorithmicComplexityAttacks}, researchers have only recently attempted to estimate its impact.
In 2018, Davis \etal reported that vulnerable regexes were present in many popular open-source software modules, and that engineers struggled to fix them~\cite{Davis2018EcosystemREDOS}; 
in 2019, they observed that these regexes displayed super-linear behavior in the built-in regex engines used in most mainstream programming languages~\cite{Davis2019LinguaFranca,Davis2019RegexGeneralizability}.
Concurrently, Staicu \& Pradel showed that 10\% of Node.js-based web services were vulnerable to \REDOS due to their use of vulnerable npm modules~\cite{Staicu2018REDOS}.
In 2022, Barlas \etal studied the impact of \REDOS in live web services~\cite{barlas2022REDOS}.
Even in non-backtracking engines, Turo{\v n}ov{\'a} \etal observed the impact of \REDOS~\cite{turonova22}.
These findings motivated further research into the \REDOS problem.





\mysection{Characterizations to Manually Detect Vulnerable Regexes}
\label{sec:modeling}
In past work, Davis \etal characterized the IA region of vulnerable regexes with anti-patterns, 
although with
a high false positive rate~\cite{Davis2018EcosystemREDOS}.
Brabrand \& Thomsen's theories precisely identify \emph{un}ambiguous regexes, but treat all others
as suspect, including both IA regexes and merely finitely ambiguous (\ie non-vulnerable) regexes~\cite{brabrand2010typed}.
\footnote{Finite ambiguity could cause \REDOS for complex regexes~\cite{turonova22} or when resources are limited (\eg  a low-power device like a Raspberry Pi).}

In contrast with Davis \etal's anti-patterns, we provide a theoretical grounding to formally capture their limitations (\cref{section:Theory}) and thus provide higher precision and recall (\cref{section:Experiment1}).
We also evaluate their usability when applied manually by humans (\cref{sec:Experiment2}). 
In contrast with Brabrand \& Thomsen's, our theory distinguishes between finite and infinite ambiguity, enabling developers to distinguish between likely-unproblematic (non-IA) and problematic (IA) regexes.

Finally, other characterizations of vulnerable regexes exist, but they were not proposed to be applied manually by humans.
Instead, they follow the models used by automatic detection tools, \eg expressed as finite automata \cite{Weideman2016REDOSAmbiguity,Wustholz2017Rexploiter} (see \cref{fig:ida}).
Contrasting with these other characterizations, we designed ours to be consumed by humans.
Our approach uses the modality of the regex language---the representation that developers understand best~\cite{bai2019exploring,wang2019exploring}.

\mysection{Tools to Automatically Detect Vulnerable Regexes}
\label{sec:detecting}
Berglund \etal defined a prioritized type of NFA to simulate a backtracking engine in Java and decide if a regex could show super-linear behavior~\cite{Berglund2014REDOSTheory}.
Weideman \etal also use a prioritized NFA to find IDA in it~\cite{Weideman2016REDOSAmbiguity,Weideman2017StaticExpressions}.
Wustholz \etal also looks for the IDA pattern in the NFA and computes an attack automaton that produces attack input strings~\cite{Wustholz2017Rexploiter} .
Liu \etal adds support for modeling and analyzing less common regex features, \eg set operations~\cite{liu21sp}. 
Li \etal prescribed five vulnerability patterns, although without theoretical validation~\cite{li2021redoshunter}. 

Others statically analyze different representations of the regex for vulnerability.
Kirrage \etal analyze an abstract evaluation tree of the regex~\cite{Kirrage2013rxxr}.
Rathnayake \etal look for exponential branching in the regex evaluation tree~\cite{Rathnayake2014rxxr2}.
Sugiyama \etal analyzes the size of a tree transducer for the regex~\cite{Sugiyama2014RegexLinearityAnalysis}. 
Finally, 
Sulzmann \etal use Brzozowski derivatives to create a finite state transducer to generate parse trees and minimal counter-examples~\cite{Sulzmann2017DerivAmbig}. 

Still other approaches detect vulnerable regexes using dynamic analysis.
Shen \etal and McLaughlin \etal proposed search algorithms to find inputs with super-linear matching time~\cite{Shen2018ReScueGeneticRegexChecker}~\cite{mclaughlin2022regulator}. 
More general algorithmic complexity detectors, \eg~\cite{Petsios2017SlowFuzz,wei2018singularity,noller2018badger,meng2018rampart,blair2020hotfuzz}, can also be extended to detect \REDOS.

Vulnerable regex detection tools have been evaluated for effectiveness, but not for usability.
Our anti-patterns complement these tools by improving developer understanding of the outcome of their detection (\cref{sec:Experiment3}).

\mysection{Tools to Automatically Fix Vulnerable Regexes}
\label{sec:repairing}
These approaches offer trade-offs for the fixed regex, in: semantic similarity, (perceived) readability, and support for uncommon features.
Van der Merwe \etal presented a modified flow algorithm to convert an ambiguous K-regex into an equivalent unambiguous one~\cite{VanDerMerwe2017EvilRegexesHarmless}, with perfect semantic equivalence, but lower readability.
More recently, Li \etal proposed an approach to fix vulnerable regexes with deterministic regex constraints to avoid regex ambiguity~\cite{li2020flashregex}.
Chida \& Terauchi proposed a ``Programming By Example'' approach that supports K-regexes, lookarounds, capture groups, and backreferences~\cite{chida2022repairing}.
Both approaches use a human in the loop to provide good examples~\cite{li2020flashregex,chida2022repairing}.
Finally, Claver \etal~\cite{claver2021regis} proposed a synthesis-based approach that they evaluate with synthetic regexes.

These tools have been evaluated for effectiveness, but not for usability.
Our fix strategies complement these tools by improving developer' understanding of the fix (\cref{sec:Experiment3}).

\mysection{Non-regex-based Workarounds}
\label{sec:workarounds}

\mysubsection{Recovering From \REDOS}
After a system containing a \REDOS vulnerability is deployed, it is possible to detect and mitigate \REDOS attacks.
Bai \etal proposed a \REDOS-specific approach, applying deep learning to detect and sandbox attack strings~\cite{bai2021runtime}.
Atre \etal proposed using adversarial scheduling to mitigate adversarial complexity attacks~\cite{atre2022surgeprotector}.
Approaches that detect anomalous resource utilization, \eg time~\cite{davis2018sense}, CPU~\cite{meng2018rampart}, or application-level concepts~\cite{demoulin2019detecting}, can also mitigate \REDOS.
These approaches reduce the impact of \REDOS, but do not remove the root cause.

\mysubsection{Changing the regex engine}
There are both classic and more recent alternatives to the exponential-time backtracking regex algorithm.
The earliest published regex matching algorithms, by Brzozowski in 1964~\cite{brzozowski1964derivatives} and Thompson in 1968~\cite{Thompson1968LinearRegexAlgorithm}, offer linear-time guarantees.
There are production-grade implementations of Thompson's approach, notably RE2~\cite{Cox2010RE2Implementation} and the engines in Rust~\cite{RustRegexDocs} and Go~\cite{GoRegexDocs}.
Microsoft has considered Brzozowski's approach for .NET~\cite{saarikivi2019symbolic}, as well as deterministic~\cite{holik2019succinct,turovnova2020regex} or hybrid~\cite{sung22} matching strategies. 
However, programming language maintainers have been slow to adopt these algorithms because of the risk of regression and the limited support for non-regular regex features~\cite{davis2021using}.

\section{Theory of Regex Infinite Ambiguity} \label{section:Theory}

Here we introduce an existing theory of regex ambiguity \cite{brabrand2010typed}, discuss its limitations, and 
present our novel theorems. 

Recalling~\cref{section:Background}, a regex with an infinite degree of ambiguity (IA)~\cite{weber1991degree} has the necessary condition for super-linear regex behavior~\cite{Davis2018EcosystemREDOS,Crosby2003REDOS,Crosby2003AlgorithmicComplexityAttacks}.
Though the NFA-level conditions for IA regexes (namely PDA and EDA regions) are well known~\cite{weber1991degree}, we lack characterizations in terms of regex syntax and semantics.
We provide such a description to support developers assessing or composing regexes.

\mysection{Preliminaries}
Brabrand \& Thomsen~\cite{brabrand2010typed} developed the state of the art description of regex-level ambiguity.
They introduced an overlap operator, $\xblnot$, between two languages $L(R_1)$ and $L(R_2)$.
The set $L(R_1) \ \xblnot \ L(R_2)$ contains the ambiguity-inducing strings that can be 
parsed in multiple ways across $L(R_1)$ and $L(R_2)$.
More formally, with $X = L(R_1)$ and $Y = L(R_2)$,
\[
X \ \xblnot \ Y = \{ xay \ | \ x,y \in {\Sigma}^{*} \wedge a \in {\Sigma}^{+} s.t. \ x,xa \in X \wedge ay,y \in Y \}
\]

\vspace{-2px}
\noindent
Using this operator, \cref{thm:Brabrand} summarizes their findings.

\begin{framedThm}[Brabrand \& Thomsen~\cite{brabrand2010typed}] \label{thm:Brabrand}
    Given unambiguous regexes $R_1$ and $R_2$:
    \begin{enumerate}[label=(\alph*)]
        \item $R_1|R_2$ is unambiguous iff $L(R_1) \cap L(R_2) = \phi$.
        \item $R_1{\cdot}R_2$ is unambiguous iff $L(R_1) \ \xblnot \ L(R_2) = \phi$.
        \item $R_{1}*$ is unambiguous iff $\epsilon \notin L(R_1) \wedge L(R_1) \ \xblnot \ L(R_{1}*) = \phi$.
    \end{enumerate}
\end{framedThm}

In their implementation of this Theorem, Brabrand \& Thomsen use M{\o}ller's BRICS library~\cite{moller2010dk}, and actually rely on what we call the M{\o}ller overlap operator, $\Omega$.
We use this operator in our theorems.
The M{\o}ller overlap operator describes only the ambiguous core ``$a$'':
\[
X \ \Omega \ Y = \{ \exists \ x, y \in {\Sigma}^{*} \wedge a \ | \ a \in {\Sigma}^{+}  s.t. \ x, xa \in X \wedge ay, y \in Y \}
\]
\vspace{-8px}





\myparagraph{Limitation}
Given unambiguous regex components, \cref{thm:Brabrand} specifies when a composed regex remains unambiguous. 
Yet not all ambiguity is harmful. 
For example, the regex \verb<\w|\d< is finitely ambiguous.
This regex formulation may improve readability~\cite{aho2020compilers}; it is not a \REDOS risk.

\mysection{Regex Infinite Ambiguity Theorems} \label{section:theory-proofs}
This section presents our regex ambiguity theory for composition with alternation (\cref{thm:alternation}), concatenation (\cref{thm:concat}), and star (\cref{thm:star}).
Here we give proof sketches, examples, and the \REDOS implications.
Full proofs are in~\cref{section:Appendix-Proofs}.



\begin{framedThm}[Ambiguity of Alternation] \label{thm:alternation}
    Given unambiguous regexes $R_1$ and $R_2$,
    \begin{enumerate}[label=(\alph*)]
        \item $R_1|R_2$ is finitely ambiguous iff $L(R_1) \cap L(R_2) \neq \phi$.
        \item $R_1|R_2$ cannot be infinitely ambiguous.
    \end{enumerate}
\end{framedThm}

\myparagraph{Proof sketch}
The theorem states that given unambiguous regexes $R_1$ and $R_2$, if $R_1|R_2$ is ambiguous, then it is always finitely ambiguous.
Since $R_1$ and $R_2$ are both unambiguous, for any matching input $w$, there is only one path through $R_1$ and $R_2$.
Therefore, for $R_1|R_2$ and any matching input $w$, there are at most two matching paths.

\myparagraph{Example} 
For regex \verb<a*|a*<, consider input ``$a...a$'' of length $N$.
Regardless of input length, the number of accepting paths will be 2:
via the first $a*$ or the second $a*$.

\myparagraph{\REDOS implications} 
If two regexes $R_1$ and $R_2$ are unambiguous, $R_1|R_2$ is always safe (cannot form IA).


\begin{framedThm}[Ambiguity of Concatenation] \label{thm:concat}
    Suppose unambiguous regexes $R_1$ and $R_2$,
    and that
      $L(R_1) \ \xblnot \ L(R_2) \neq \phi$
    (so $R_1{\cdot}R_2$ is ambiguous by~\cref{thm:Brabrand}).
    Then:
    \begin{enumerate}[label=(\alph*)]
        \item $R_1{\cdot}R_2$ is infinitely ambiguous iff 
        $L(R_1)$ contains the language of a regex \verb<BC*D< and 
        $L(R_2)$ contains the language of a regex \verb<EF*G<, 
        where $\epsilon \notin L($\verb<C<$) \wedge \epsilon \notin L($\verb<F<$) \wedge L($\verb<C<$) \cap L($\verb<F<$) \cap L($\verb<DE<$) \neq \phi$.
        \item Otherwise, $R_1{\cdot}R_2$ must be finitely ambiguous.
    \end{enumerate}
\end{framedThm}

\myparagraph{Proof sketch}
$\impliedby$:
Consider the string ``\emph{bcc...cdeff...fg}'' $\in L(R_1 \cdot R_2)$ where $c=f=de$.
It can be divided into two strings ``\emph{bcc...cd}'' $\in L($\verb<BC*D<$) \subseteq L(R_1)$ and ``\emph{eff...fg}'' $\in L($\verb<EF*G<$) \subseteq L(R_2)$.
By hypothesis, we can repeat the substring ``\emph{de}'' arbitrarily many times, and the resulting string can be matched in $R_1$ (by \verb<C*<) or in $R_2$ (by \verb<F*<).
We can choose an arbitrarily long string and obtain arbitrary ambiguity in $R_1{\cdot}R_2$.

$\implies$:
Suppose $R_1{\cdot}R_2$ is infinitely ambiguous.
The NFA corresponding to $R_1{\cdot}R_2$ cannot contain the EDA structure because this requires a self-loop --- \ie that $R_1$ or $R_2$ is already ambiguous. Therefore the NFA of $R_1{\cdot}R_2$ must contain a PDA structure, as shown in~\cref{fig:ida}(a).
We can map the two loops $\pi_1$ and $\pi_3$ with \verb<C*< and \verb<F*< respectively; and the bridge $\pi_2$ with \verb<DE< in the regex representation,  where $L($\verb<C<$) \cap L($\verb<F<$) \cap L($\verb<DE<$) \neq \phi$.

\myparagraph{Example}
For regex \verb<(a*a)(aa*)< on input ``$aa...a$'' of length $N$, 
there are $N$ accepting computations, one for each of the indices of the input dividing the string into a left half consumed by $R_1$ and a right half consumed by $R_2$.

\myparagraph{\REDOS implications}
Though two regexes $R_1$ and $R_2$ are unambiguous, $R_1{\cdot}R_2$ could be IA, thus concatenation should be used with care.
\cref{thm:concat}(a) implies that for $R_1{\cdot}R_2$ to be IA, there must be a star component in both $R_1$ and $R_2$.
In~\cref{section:AntiPatterns:concat}, we introduce three forms of concatenation anti-patterns based on this observation.

\begin{framedThm}[Ambiguity of Star] \label{thm:star}
    Given unambiguous regex R, 
    \begin{enumerate}[label=(\alph*)]
        \item $R*$ is infinitely ambiguous iff $ \epsilon \in L(R) \vee L(R) \ \Omega \ L(R*) \neq \phi$.
        \item $R*$ cannot be finitely ambiguous.
    \end{enumerate}
\end{framedThm}

\myparagraph{Proof sketch}
%
The theorem states that given an unambiguous regex $R$, if $R*$ is ambiguous, then it is always infinitely ambiguous.
Suppose $R*$ is ambiguous.
Then there is some input $w$ that it can match in $k$ ways, $k > 1$.
So there is an input $ww$ that it can match in $k*k = k^2$ ways.
The degree of ambiguity increases as a function of input length.

\myparagraph{Example}
For the regex \verb<(a*)*<, consider input ``$aaa...a$'' of length $N$.
There are two ways (inner \verb<*< or outer \verb<*<) to match each ‘$a$’, making the total number of ways to match to be $2^N$.

\myparagraph{\REDOS implications} 
Even though an original regex $R$ is unambiguous, $R*$ can be IA.
In~\cref{section:AntiPatterns:star}, we give an anti-pattern that only checks for a subset of conditions for simplicity.

\begin{framedThm}\label{thm:star-lemma}
    Given a finitely ambiguous regex $R$, $R*$ is always infinitely ambiguous.
\end{framedThm}

\myparagraph{Proof sketch}
The proof follows the logic of~\cref{thm:star}.

\myparagraph{Example}
For the regex \verb<(a|a)*<, consider an input ``$aaa...a$'' of length $N$.
There are two ways (first \verb<a< or second \verb<a<) to match each ‘$a$’ of the input, for $2^N$ matches in all. 

\myparagraph{\REDOS implications} 
If $R$ is finitely ambiguous, from alternation ($R=P|Q$) or concatenation ($R=P{\cdot}Q$), $R*$ is always IA. 
Later in~\cref{section:AntiPatterns:star}, we introduce two anti-patterns of the form $(P|Q)*$.




\section{Anti-patterns for Regex Infinite Ambiguity} 
\label{section:AntiPatterns}






This section describes anti-patterns for IA regexes (\emph{IA anti-patterns}), derived from the preceding theory of regex infinite ambiguity.
Ideal anti-patterns would be as sound and complete as the theory, but this goal must be balanced against usability. 
With this in mind, we iteratively extracted IA anti-patterns from the theory by dropping clauses from theorems or combining the theorems in different ways.
These anti-patterns were refined through internal discussion.
We evaluate these anti-patterns in~\cref{section:Experiment1} and~\cref{sec:Experiment2}.

\cref{tab:our_antipatterns} summarizes our IA anti-patterns.
As alternation alone does not make a regex IA (\cref{thm:alternation}),
  there are \emph{Concatenation} anti-patterns derived from~\cref{thm:concat},
  and \emph{Star} anti-patterns derived from~\cref{thm:star,thm:star-lemma}. 

\begin{table*}[t]
\centering
\small
\caption{
  Our proposed IA anti-patterns.
  Each row indicates
    the anti-pattern,
    the theorem(s) from which it was derived,
    a description,
    and
    an example of how the anti-pattern leads to ambiguity.
}
\begin{tabular}{@{} p{0.10\linewidth} @{} p{0.08\linewidth} @{} p{0.44\linewidth} p{0.34\linewidth} @{}}
\toprule
\textbf{Anti-pattern} & \centering\textbf{Thm.} & \textbf{Description}                                                                          & \textbf{Example}                                                         \\ \midrule
Concat 1              & \centering 2 &
\verb<R< = \verb<...P*Q*...< (\verb<R< has a sub-regex \verb<P*Q*<) --- The two quantified parts \verb<P*< and \verb<Q*< can match some shared string $s$. & 
\verb<\w*\d*< --- both classes can match digits [0-9].
\\ \midrule
Concat 2              & \centering 2 &
\verb<R< = \verb<...P*SQ*...<  --- The two quantified parts \verb<P*< and \verb<Q*< can match a string $s$ from the middle part $S$. & 
\verb<\w*0\d*< --- the repeated classes \verb<\w< and \verb<\d< can match the middle part $0$. 
\\ \midrule
Concat 3              & \centering 2 &
\verb<R< = \verb<...P*S*Q*...< --- Advanced form of Concat 1. 
Since \verb<S*< includes an empty string, the ambiguity between \verb<P*< and \verb<Q*< can be realized. & 
\verb<\w*:*\d*< --- The classes \verb<\w< and \verb<\d< overlap, and the intervening \verb<:*< can be skipped. 
\\ \midrule
Star 1                & \centering 1, 4 &
\verb<R*<, \verb<R<= \verb<(P|Q|...)< ---
There is an intersection between any two alternates, \ie both match some shared strings. & 
\verb<(\w|\d)*< --- both classes match digits [0-9]. 
\\ \midrule
Star 2                & \centering 3 &
\verb<R*<, \verb<R<= \verb<(P|Q|...)< --- You can make one option of the alternation by repeating another option multiple times or by concatenating two or more options multiple times. &
\verb<(a|b|ab)*< --- The 3rd option, $ab$, matches combinations of the first and second options.
\\ \midrule
Star 3                & \centering 3 &
\verb<R*<, \verb<R<= \verb<(...P*...)< --- Nested quantifiers, provided \verb<RR< follows any of the Concat anti-patterns. & 
Expanding \verb<R=(0?\w*)*< to \verb<RR< yields \verb<0?\w*0?\w*<, which is IA by Concat 3.
Similarly, \verb<R=(xy*)*< yields \verb<xy*xy*<; this is not IA by any Concat anti-pattern. 
\\ \bottomrule
\end{tabular}
\label{tab:our_antipatterns}
\end{table*}

\mysection{Concatenation Anti-patterns} 
\label{section:AntiPatterns:concat}
The Concat anti-patterns come from~\cref{thm:concat}.
In~\cref{thm:concat}, a regex $R$ concatenates regexes $R_1$ and $R_2$ that contain the languages \verb|BC*D| and \verb|EF*G|, respectively. 
The theorem states that the potential vulnerability occurs in the sub-regex \verb|C*DEF*|, which we write in simplified form as \verb|P*SQ*| for our anti-patterns.
We call \verb|S| the ``bridge'' between \verb|P*| and \verb|Q*|.

\myitparagraph{Concat-1}
This anti-pattern,
  where $P*Q*$ is a sub-regex of $R$,
represents the simplest form without the bridge $S$.
Developers must find a string matched in both $P*$ and $Q*$.

\myitparagraph{Concat-2}
This anti-pattern,
  where $P*SQ*$ is a sub-regex of $R$,
has the bridge $S$ component.
Developers must find a string matched in all $P*$, $Q*$, and $S$.

\myitparagraph{Concat-3}
This anti-pattern,
  where $P*S*Q*$ is a sub-regex of $R$,
is the case with optional bridge $S$.
Like Concat-1, developers must find a string matched in both $P*$ and $Q*$.

\myparagraph{Gap Analysis} 
The Concat anti-patterns represent all possible ways that the bridge component $DE$ (from \cref{thm:concat}) may appear as a sub-regex of the form $\epsilon$, $S$, or $S*$.
Thus, there is no gap between theory and anti-patterns.


\mysection{Star Anti-patterns}
\label{section:AntiPatterns:star}
The Star anti-patterns come from~\cref{thm:star} and~\cref{thm:star-lemma}. 

\myitparagraph{Star-1 and Star-2} These anti-patterns are designed to prevent (some) regexes of the form $R*$ where $R=(P|Q|...)$.
\cref{thm:star-lemma} states that if $R$ is finitely ambiguous, then $R*$ becomes IA.
From~\cref{thm:alternation}(a), alternations may introduce finite ambiguity.
The Star-1 anti-pattern describe the condition when the subregex $(P|Q|...)$ becomes finitely ambiguous. The Star-2 anti-pattern describe the condition when the non-ambiguous $(P|Q|...)$ form IA with the help $*$ according to~\cref{thm:star}.

\myparagraph{Gap Analysis}
There is a gap between~\cref{thm:star-lemma} and the Star-1 anti-pattern.
Star-1 does not consider all possible forms of finitely ambiguous regexes. 
For instance, the concatenation may also introduce finite ambiguity (\cref{thm:concat}(b)).
Thus, some regexes of the form $(P{\cdot}Q)*$ could be IA as well: \eg \verb<((a|ab)(c|bc))*<.
 Also, the Star-2 anti-pattern is one of the conditions that incorporate~\cref{thm:star}. Thus regexes under missing conditions would appear as false negatives for these anti-patterns.

\myitparagraph{Star-3} This anti-pattern prevent (some) regexes of the form $R*$ where $R$ has a sub-regex $P*$.
\cref{thm:star} states the conditions when $R*$ becomes IA.
considering the first condition $ \epsilon \in L(R)$ is relatively trivial.
Yet, the second condition $L(R) \Omega L(R*) \neq \phi$ requires reasoning about a language overlap between $L(R)$ and an arbitrary repetition of $L(R*)$, which could be tricky. 
Based on the common knowledge that a nested quantifier (\eg $(P*)*$) is bad~\cite{safeRegexHomePage},
the Star-3 anti-pattern only considers the case where $R$ has a sub-regex $P*$, as a generalized form of nested quantifiers.
The Star-3 anti-pattern further simplifies the condition and asks developers to consider the overlap between $L(R)$ and (twice-repeated) $L(R{\cdot}R)$, using the Concat anti-patterns.

\myparagraph{Gap Analysis}
The Star-3 anti-pattern does not incorporate all the conditions in \cref{thm:star}.
Regexes with the missing conditions would appear as false negatives. 

\section{Fix Strategies for Regex Infinite Ambiguity}
\label{sec:fix_strategies}


This section describes five fix strategies (F1--F5) that can be broadly applied across the different IA anti-patterns.
The fix strategies are derived from various ways of invalidating necessary conditions of~\cref{thm:concat} and~\cref{thm:star}.
The proposed fix strategies do not always preserve semantics.

\mysection{Fix strategies}
\cref{tab:fix_strategies} summarizes the proposed five fix strategies along with examples for each anti-pattern.
We evaluate their effectiveness in Experiment 3 (\cref{sec:Experiment3}).

\begin{table*}[]
\centering
\caption{
  Fix strategies.
  Each strategy is illustrated with respect to each anti-pattern, within the limit of the example provided. 
}
\label{tab:fix_strategies}
\resizebox{0.99\textwidth}{!}{%


\begin{tabular}{@{}l@{\enspace}l@{}c@{\enspace}c@{\enspace}c@{\enspace}c@{\enspace}c@{\enspace}c@{\enspace}c@{\enspace}c@{}}
\toprule
\multirow{3}{*}{\textbf{Fix}} &
  \multirow{3}{*}{\textbf{Description}} &
  \textbf{Concat 1}  &
  \textbf{Concat 2}  &
  \textbf{Concat 3}  &
  \textbf{Star 1} &
  \textbf{Star 2} &
  \textbf{Star 3} & 
  \multirow{3}{*}{\textbf{Freq.}}  \\
 &
   &
  \emph{Anti-pattern}: ...P*Q*... &
  ...P*SQ*... &
  ...P*S*Q*... &
  (P\textbar{}Q\textbar{}...)* &
  (P\textbar{}Q\textbar{}...)* &
  (...P*...)* & 
   \\
 &
   &
  \emph{Example}: \ \ \ \textbackslash{}w*\textbackslash{}d* &
  \textbackslash{}w*0\textbackslash{}d* &
  \textbackslash{}w*:*\textbackslash{}d* &
  (\textbackslash{}w\textbar{}\textbackslash{}d)* &
  (a\textbar{}b\textbar{}ab)* &
  (0?\textbackslash{}w*)* & 
  \\ \midrule
F1 &
  \begin{tabular}[c]{@{}l@{}}Add a delimiter between the \\ sub-regexes P and Q that \\ can match a shared string.\end{tabular} &
  \textbackslash{}w*:\textbackslash{}d* &
  \textbackslash{}w*:0\textbackslash{}d* &
  \textbackslash{}w*:+\textbackslash{}d* &
  (\textbackslash{}w*\textbar{}:\textbackslash{}d)* & (a:$|$b$|$ab)*
   &
  (:0?\textbackslash{}w*)* & 12 \\ \midrule
F2 &
  \begin{tabular}[c]{@{}l@{}}Reduce one of the sub-regexes \\ P and Q so that it no longer \\ matches any of the shared strings.\end{tabular} &
  {[}a-zA-Z\_{]}*\textbackslash{}d* &
  {[}a-zA-Z\_{]}*0\textbackslash{}d* &
  {[}a-zA-Z\_{]}*:*\textbackslash{}d* &
  ({[}a-zA-Z\_{]}\textbar{}\textbackslash{}d)* & (b$|$ab)*
   & (0?{[}a-zA-Z\_{]})* & 10 \\ \midrule
F3 &
  \begin{tabular}[c]{@{}l@{}}Reduce both the sub-regexes \\ P and Q so that they no longer \\ match any of the shared strings, \\ and add the shared string(s) in a \\ disjunction. In many cases, this \\ will resemble making a superset \\ of the two sub-regexes.\end{tabular} &
  \textbackslash{}w* &
  N/A &
  N/A &
  \textbackslash{}w* & (a$|$b)*
   &
  \textbackslash{}w* & 3 \\ \midrule
F4 &
  \begin{tabular}[c]{@{}l@{}}Reduce or remove repetition in at \\ least one of the sub-regexes 
  \\ P and Q that match a shared string.\end{tabular} &
  \textbackslash{}w\{,10\}\textbackslash{}d\{,10\} &
  \textbackslash{}w\{,10\}0\textbackslash{}d\{,10\} &
  \textbackslash{}w\{,10\}:*\textbackslash{}d\{,10\} &
  (\textbackslash{}w\textbar{}\textbackslash{}d)\{,10\} & (a$|$b$|$ab)\{,10\}
   &
  (0?\textbackslash{}w\{,10\})\{,10\} & 13 \\ \midrule
F5 &
  \begin{tabular}[c]{@{}l@{}}Remove or substantially modify \\ the sub-regexes P and Q and \\ add logic for 
  the semantic changes. \end{tabular} &
  \begin{tabular}[c]{@{}c@{}}add logic to \\ catch non-digits \\ then use \textbackslash{}d*\end{tabular} &
  \begin{tabular}[c]{@{}c@{}}add logic to \\ catch non-digits \\ then use 0\textbackslash{}d*\end{tabular} &
  \begin{tabular}[c]{@{}c@{}}add logic to \\ catch non-digits \\ then use :*\textbackslash{}d*\end{tabular} &
  \textbackslash{}w*\textbar{}\textbackslash{}d* & a+$|$b+$|$(ab)+
   & \begin{tabular}[c]{@{}c@{}}add logic to \\ catch 0, then\\ use \textbackslash{}w*\end{tabular}
   & 16 \\ 
   \bottomrule
\end{tabular}

}
\end{table*}

\begin{enumerate}

\item The first fix strategy (F1) is to add a delimiter between the subregexes \verb<P< and \verb<Q< of the anti-patterns (\eg \verb<P*Q*<, \verb<(P|Q)*<) that can match the shared string(s).
More precisely, the delimiter makes $L(C) \cap L(DE) = \phi$ and/or $L(F) \cap L(DE) = \phi$ in~\cref{thm:concat}(a); and $L(R) \ \Omega \ L(R*) = \phi$ in~\cref{thm:star}(a). 
For instance, consider the regex \verb<\w*\d*< (Concat 1). If we add a delimiter `:', 
the new regex \verb<\w*:\d*< becomes non-IA because $L(\verb<\w<) \cap L(\verb<:<) = \phi$ and $L(\verb<\d<) \cap L(\verb<:<) = \phi$. 
\cref{tab:fix_strategies} provides examples for the other anti-patterns.

\item The second fix strategy (F2) is to reduce one of the subregexes \verb<P< and \verb<Q< so that it no longer matches any of the shared strings.
In other words, the fix F2 makes $L(C) \cap L(F) = \phi$ in~\cref{thm:concat}(a); and $L(R) \ \Omega \ L(R*) = \phi$ in~\cref{thm:star}(a). 
For instance, refer to the same regex \verb<\w*\d*< (Concat 1). 
If we reduce \verb<\w< to \verb<[A-Za-z_]< so that it does not overlap with \verb<\d<,
the new regex \verb<[A-Za-z_]*\d*< is not IA since $L(\verb<[A-Za-z_<]) \cap L(\verb<\d<) = \phi$.

\item The third fix strategy (F3) is to reduce the subregexes \verb<P< and \verb<Q< so that they no longer match any of the shared strings, and add the shared string(s) in a disjunction. In many cases, this will resemble making a superset of the two sub-regexes. 
Effectively, the fix F3 has the same effect as F2 that excludes any shared string(s), yet it additionally keeps the shared string(s) in a disjunction, making the fix semantic-preserving.
For example, consider \verb<(\w|\d)*< (Star 1).
Suppose we reduce \verb<\w< to \verb<[A-Za-z_]< and reduce \verb<\d< to null so that they no longer match the shared string(s) \verb<[0-9]<.
Then we add the subregex \verb<[0-9]< in a disjunction.
Finally, we get \verb<([A-Za-z_]|[0-9])*<, which is equivalent to \verb<\w*<.
Note that \verb<\w< is a superset of \verb<\w< and \verb<\d<, and the old and new regexes match the same language.

\item The fourth fix strategy (F4) is to reduce or remove repetition in at least one of the subregexes P and Q so that the shared string(s) cannot be matched infinitely.
Both~\cref{thm:concat} and~\cref{thm:star} require an unbounded repetitions (a star quantifier).
The fix F4 in effect turns unbounded repetitions to bounded ones.
For example in Star 1 anti-pattern, we can replace the regex \verb<(\w|\d)*< with \verb<(\w|\d){0,10}< permitting only up to 10 repetitions.

\item The fifth fix strategy (F5) is to remove or substantially modify the subregexes P and Q and handle semantic changes elsewhere. 
The fix F5 capture general non-systematic fixes that may introduce more semantic changes than the other fixes.
For example, the regex \verb<(\w|\d)*< can be fixed to \verb<\w*|\d*<.

\end{enumerate}

\myparagraph{Evaluation of Practical Relevance}
\label{sec:fix_strategies:popularity}
We analyzed the 54 
developer-created regex fixes reported by Davis \etal~\cite{Davis2018EcosystemREDOS}.
We classified each fix into one of these five strategies.
The last column in \cref{tab:fix_strategies} reports the frequency of each fix strategy.
We also observed that developers value simplicity in the fix.
To preserve the original (vulnerable) regex's structure, they introduced semantic changes (51/54 = 94\%).
\section{Experiment 1: Effectiveness of Anti-patterns}
\label{section:Experiment1}
Our proposed IA anti-patterns (\cref{section:AntiPatterns}) were derived from our theory (\cref{section:Theory}), but we deliberately introduced inaccuracy in favor of simplicity.
In this section, we evaluate the impact of these deviations 
over the largest available regex corpus~\cite{Davis2019LinguaFranca}.
We measured effectiveness using precision and recall.

\subsection{Experimental Design} \label{sec:Experiment1-Design}


\mysection{Studied Techniques}
\label{section:Experiment1:techniques}
\emph{Our IA anti-patterns.}
We detected each anti-pattern using static analysis. 
We parsed regexes in PCRE format~\cite{Friedl2002MasteringRegexes} using an ANTLR 4 grammar and parser~\cite{ANTLRPCRE}.
We used the BRICS~\cite{moller2010dk} tool to check whether multiple sub-regex parts can generate any shared string.

We implemented our IA anti-patterns to support the common case of K-regexes, \ie regexes that use only Kleene-regular regexes (cf. \cref{section:background:regex:regexes}).
Our prototype also excludes extended POSIX and Unicode character classes for simplicity.
These limitations are consistent with 
  past approaches~\cite{Berglund2014REDOSTheory, Weideman2016REDOSAmbiguity, Weideman2017StaticExpressions, Wustholz2017Rexploiter, Kirrage2013rxxr, Rathnayake2014rxxr2, Sugiyama2014RegexLinearityAnalysis, Sulzmann2017DerivAmbig}. 



\emph{State-of-the-art (SOA) anti-patterns.}
For comparison, we also executed the state-of-the-art (SOA) anti-patterns that characterize IA regexes~\cite{Davis2018EcosystemREDOS}.
We used the automatic detector provided by Davis \etal~\cite{Davis2018EcosystemREDOS}.
These anti-patterns lack a theoretical basis, so we expect them to perform worse.
Davis \etal described three anti-patterns, listed in \cref{tab:davis_antipatterns}\footnote{As might be expected, these anti-patterns resemble those presented in~\cref{tab:our_antipatterns}. The main difference is in the nuanced definition of ``overlap'' available from our IA theory.}

\begin{table}[t]
\centering
\small
\caption{
  State of the art IA anti-patterns, as described by Davis \etal \cite{Davis2018EcosystemREDOS}.
  Each row indicates
    the anti-pattern,
    its description.
    and
    an example.
}
\begin{tabular}{@{} p{0.25\linewidth} @{} p{0.7\linewidth}}
\toprule
\textbf{Anti-pattern} & \textbf{Description}                                                         
\\ \midrule
QOA (Quantified Overlapping Adjacency) \newline \newline Example: \verb</\w*#?\w*/<             &
\emph{The two quantified} \verb<\w*< \emph{nodes overlap, and are adjacent because one can be reached from the other by skipping the optional octothorpe. From each node we walk forward looking for a reachable quantified adjacent node with an overlapping set of characters, stopping at the earliest of: a quantified overlapping node (QOA), a non-overlapping non-optional node (no QOA), or the end of the nodes (no QOA).}
\\ \midrule
QOD (Quantified Overlapping Disjunction) \newline \newline Example: \verb</(\w|\d)+/<              &
\emph{Here we have a quantified disjunction} \verb<(/(...|...)+/)<, \emph{whose two nodes overlap in the digits, 0-9.}
\\ \midrule
Star height \textgreater 1 \newline \newline Example: \verb</(a+)+/<              &
\emph{To measure star height, we traverse the regex and maintain a counter for each layer of nested quantifier: +, *, and check if the counter reached a value higher than 1. In such cases, the same string can be consumed by an inner quantifier or the outer one, as is the case for the string ``a'' in the regex} \verb</(a+)+/<.
\\ \bottomrule
\end{tabular}
\label{tab:davis_antipatterns}
\end{table}

\mysection{Ground Truth} \label{section:Experiment1:groundtruth}
We assessed ground truth for whether a regex is IA using 
Weideman \etal's detector~\cite{Weideman2017StaticExpressions,Weideman2016REDOSAmbiguity}.
This detector tests if a regex is IA by analyzing its NFA (\cref{fig:ida}).
Since the Weideman tool uses automata theory instead of regex semantic theory, it provides an independent check on the anti-patterns (and our underlying theory).


\mysection{Metrics}
The standard metrics of precision and recall~\cite{ting10}.

\mysection{Dataset} \label{section:Experiment1:dataset}
We evaluated the studied anti-patterns in the largest available dataset of real-world regexes~\cite{Davis2019LinguaFranca} (\CorpusNRegexes regexes). 
This dataset has been used by previous studies for
  measuring \REDOS~\cite{Davis2019LinguaFranca},
  fixing \REDOS~\cite{davis2021using,chida2020automatic},
  and
  measuring general characteristics of regexes~\cite{Davis2019RegexGeneralizability}.
We analyzed \AnalyzedNRegexes regexes from this dataset
--- one order of magnitude larger than the evaluation of past \REDOS-detection approaches (15,000--30,000 regexes~\cite{liu21sp,Shen2018ReScueGeneticRegexChecker,Weideman2016REDOSAmbiguity,Weideman2017StaticExpressions}).

We curated the dataset for this experiment:
(1) We removed the 295,151 regexes that were not supported by the ground truth tool~\cite{Weideman2017StaticExpressions}).\footnote{While our implementation of our IA anti-patterns supported a larger percentage of the dataset, we discarded those for which we could not collect ground truth.}
According to our ground truth, 32,005 of our studied regexes were IA.
(2) We discarded 32,413 additional regexes that were not supported by the BRICS library~\cite{moller2010dk} used in our anti-pattern prototype.
(3) We discarded 1,054 additional regexes with POSIX or Unicode character classes not supported by our implementation.

The implementation of our IA anti-patterns supports 450,753 regexes (83.8\%) of the dataset (10.2\% had advanced or non-regular features; 6\% unsupported by BRICS). 
This level of completeness is comparable to prior research prototypes for regex analysis \cite{davis2021using,Rathnayake2014rxxr2,Wustholz2017Rexploiter}.

Finally, we measured regex generalizability metrics~\cite{Davis2019RegexGeneralizability} in the regexes that we kept and filtered out. We found that they were similar in median: length (18 vs. 19), paths (1 vs. 1), features (3 vs. 4), and ratio of IA regexes flagged by our anti-patterns (15\% vs. 15.8\%).





\begin{table}[]
\centering
\small
\caption{
  Comparison of Precision and Recall between our IA anti-patterns and SOA anti-patterns.
}
\begin{tabular}{ccc}
\toprule
\textbf{Anti-patterns}         & \textbf{Precision} & \textbf{Recall} \\ \toprule
Our IA anti-patterns               & 100\%               & 99\%            \\
SOA anti-patterns~\cite{Davis2018EcosystemREDOS}              & 50\%               & 87\%            \\
\bottomrule
\end{tabular}
\vspace{-15px}
\label{tab:performance_comparision}
\end{table}

\subsection{Results} \label{section:Experiment1:results}

\mysubsection{How Effective were Our IA Anti-patterns Compared to the SOA Anti-patterns?}
In Table~\ref{tab:performance_comparision}, we report the results for our studied anti-pattern families.
It shows that our proposed IA anti-patterns provided a substantial improvement in both precision ($100\%$ compared to $50\%$) and recall ($99\%$ compared to $87\%$) when compared to the SOA anti-patterns.
Our IA anti-patterns addressed many of the false positives of the SOA.
For example, the \emph{Star height $>1$} anti-pattern can produce many false positives \eg the non-IA regex \verb</(b*c)*/< has \emph{Star height = 2}.
Our IA anti-patterns also reduced the number of false negatives of the SOA, \eg for regexes like \verb<(a|b)*(ab)*< and \verb<(a|b|ab)*<.
The SOA anti-patterns find no overlap between $(a|b)$ and $(ab)$ and would not label them as IA. 
In contrast, our \emph{Concat 1} and \emph{Star 2} anti-patterns, respectively, would label both as IA. 

We note that we observed higher precision and recall achieved by the SOA anti-patterns than was reported by their original study~\cite{Davis2018EcosystemREDOS}. 
We suggest two reasons: we studied a different dataset,
and we assumed full match for unanchored regexes (\eg converting \verb<a+< to \verb</^.*?a+$/<)~\cite{Davis2019LinguaFranca}), which reveals more IA regexes in the dataset.

Finally, we also performed a deeper investigation into the root cause of the false negatives of our IA anti-patterns (the 1\% of IA regexes that they did not flag as IA).
The false negatives in our experiment were mainly regexes with constructions that were too complex for our current anti-pattern scripts to detect for the limitation of Star anti-patterns discussed in~\cref{section:AntiPatterns}.
While this limitation of our implementation caused a few false negatives (affecting only 1\% of IA regexes), our implementation is still sound for our studied dataset --- it caused no false positives.

\begin{table}
\centering
\small
\caption{
  Prevalence of each of our proposed IA anti-patterns within the studied dataset.
  As some regexes fit multiple IA anti-patterns, the final row eliminates double-counting.
}
\begin{tabular}{crr}
\toprule
\textbf{IA Anti-pattern} & \textbf{\# Regexes} & \textbf{Prevalence} \\
\toprule
Concat 1                 & 17,349   & 54\%                \\
Concat 2                 & 12,419   & 39\%                \\
Concat 3                 & 414      & 1\%                 \\
Star 1                   & 192      & \textless 1\%       \\
Star 2                   & 639      & 2\%                 \\
Star 3                   & 1,133    & 4\%                 \\ \midrule
All anti-patterns        & 31,537   & 99\%                \\
\bottomrule
\end{tabular}
\vspace{-15px}
\label{tab:per_rule}
\end{table}

\mysubsection{How Prevalent was each of our IA Anti-patterns?}
Table~\ref{tab:per_rule} shows the prevalence of each of our proposed IA anti-patterns in our studied dataset, \ie the ratio of IA regexes that were detected by each IA anti-pattern.
Note that the prevalence ratios do not add up to 100\%, since some IA regexes may contain multiple anti-patterns.

We make multiple observations in this table.
First, all our IA anti-patterns as a group provided high recall
(99\%); false negatives are rare.
Second, we observed wide variations in the prevalence of each individual anti-patterns.
This means that our set of anti-patterns could be further simplified and still obtain very high recall altogether.
Somebody wanting to learn only a single anti-pattern could learn only Concat 1 and still cover 54\% of IA regexes --- adding Concat 2, one would cover the large majority ($>90\%$) of IA regexes, and so on.
This confirms past research that found that polynomial regexes were much more prevalent than exponential ones~\cite{davis2018sense}.
We believe that this does not mean that developers are already good at avoiding some anti-patterns, but instead that the kinds of problems that would require a Concat regex are more common than those that would require a Star one.
However, future work would be needed to answer this question.
Finally, before considering ignoring the less common (lower prevalence) IA anti-patterns, one should also consider their risk.
While the star anti-patterns are less common (about 6\% of all IA regexes), they are riskier --- our theory shows that they lead to exponential ambiguity.

\mysection{Summary for Experiment 1} 
Our IA anti-patterns correctly identified IA regexes with substantially higher effectiveness (100\% precision, 99\% recall) than the SOA anti-patterns (50\% precision, 87\% recall).

\section{Experiment 2: Effectiveness when Applied by Humans}
\label{sec:Experiment2}

Our IA anti-patterns can identify IA regexes with high precision and recall (\cref{section:Experiment1}), but their effectiveness may be reduced when applied manually~\cite{johnson2013don}.
Here we report on a human-subjects experiment evaluating the effectiveness of our IA anti-patterns when applied manually.

\subsection{Experimental Design} \label{sec:Experiment2-design}


\mysection{Overview}
\label{section:Experiment2:process}
We asked 20 software developers to perform 5 regex composition tasks.
To study a context in which developers may prefer to apply IA anti-patterns manually~\cite{johnson2013don}, we studied simple regex composition tasks.
We followed a within-subjects approach: each subject applied both our IA anti-patterns and the state-of-the-art (SOA) ones.
Among 20 participants,
  half (10) used our anti-patterns first,
  and the other half used the SOA ones first.
We measured whether subjects correctly identified IA in their regexes.

\mysection{Treatments}
\label{section:Experiment2:techniques}
We showed subjects our IA anti-patterns as described in~\cref{tab:our_antipatterns}
and the SOA anti-patterns using verbatim text from Davis \etal's original description of the anti-pattern, and of how it should be applied (described in~\cref{tab:davis_antipatterns}).
Note that we did not study a control group that used no anti-patterns.
Experiment 1 already answers what a control group would show: when developers are given no support, they write thousands of vulnerable regexes (\cref{section:Experiment1}).

\mysection{Tasks}
\label{section:Experiment2:tasks}
\cref{tab:experiment3_tasks} shows the 5 tasks of Experiment 2. 
Task 1 was an easy warm-up task, to familiarize subjects with the structure of the experiment.
The next three tasks (Tasks 2, 3, 4) evaluated limitations that we identified in the three SOA anti-patterns, to learn if our IA anti-patterns were more effective in those scenarios.
In task 2, \textit{Star height $>1$} may produce a false positive, assessing the regex as IA. 
In task 3, \textit{QOA} may produce a false negative, assessing the regex as non-IA. 
In task 4, \textit{QOD} may produce a false negative.
Finally, task 5 evaluated a scenario in which the SOA anti-patterns are successful, to learn if our IA anti-patterns are comparable in such a scenario.
We expected both sets of anti-patterns to perform equally in tasks 1 and 5, and our IA anti-patterns to be more effective in tasks 2, 3, and 4.

\mysection{Within-Subjects Protocol}
\label{section:Experiment2:protocol}
Our protocol had three steps:
(1) \emph{Training:} 
We shared background information in:
  (\emph{i}) regex syntax and useful terminology, so that they knew correct regex syntax;
  and
  (\emph{ii}) regex ambiguity and \REDOS, so that they understood the practical utility of the task and thus increase their engagement.
(2) \emph{First set of anti-patterns:}
We taught subjects one set of anti-patterns.
They completed the 5 regex composition tasks, producing a regex that is not IA, using the given anti-patterns.
(3) \emph{Second set of anti-patterns:} 
We taught subjects the other set of anti-patterns.
They performed the same 5 tasks, in the same order, using the other anti-patterns. 

We let subjects ask clarifying questions.
We asked them to think aloud.
The experiment took $\sim$1 hour per subject.
Subjects were compensated with a \$15 gift card.


\begin{table}[]
\centering
\small
\renewcommand{\tabcolsep}{4pt}
\caption{
    Regex composition tasks studied in Experiment 2.
    }
\begin{tabular}{cp{5.3cm}p{2cm}}
\toprule
\textbf{Task} & \textbf{Description} & \begin{tabular}{@{}c@{}}\textbf{Typical} \\ \textbf{solution}\end{tabular}  \\ \toprule
1       & Write a regex to match one or more `b' followed by a single `c'. Example matching strings: bc, bbc, bbbbc, bbbbbc, bbbbbbbbbbbbbbbbbbbbbc                                                                                                                                                & non-IA: \hspace{0.5cm} \verb<b+c<                                                                                        \\ \midrule
2       & Write a regex to match one or more repetitions of the following:  one or more `b' followed by a single `c'. Example matching strings: bcbc, bbcbbcbbc,  bbbbbcbbbbbc, bbbbbbbbbbbbbbbbbbbbbcbbbbbbbbbbbbbbbbbbbbbc & non-IA: \verb<(b+c)+<                                                                                    \\ \midrule
3       & Write a regex to match one or more `a' or `b', followed by one or more repetitions of `ab'. Example matching strings: aab, bab, aaab, aaaaab, bab, bbbab, aaaabababab, bbbbababababab                                                              & IA: \verb<(a|b)+(ab)+<                                                                              \\ \midrule
4       & Write a regex to match one or more occurrences of the strings `a', `b', or `ab'. Example matching strings: aaaaaaaaaa, bbbbbbbbbbbb, ababababababababab                                      & IA: \verb<(a|b|ab)+<                                                                               \\ \midrule
5       & Write a regex to match one or more `a' followed by an optional `b' followed by one or more `a'. Example matching strings: aaaabaa, aaaaa, abaaaa                                                                                    & IA: \verb<(a+b?a+)<                                                                                   \\ \bottomrule
\end{tabular}
\label{tab:Regex Composition Tasks}
\end{table}

\mysection{Subjects}
\label{section:Experiment2:subjects}
Subjects were recruited via posts on Twitter, Reddit (r/regex), and our institutional mailing lists.
We asked subjects to report their years of professional software development and their experience with regexes (self-reported, based on popular regex features following Michael \etal~\cite{Michael2019RegexesAreHard}).
We had 27 respondents, and kept the 21 respondents who reported some experience in both categories.
After performing the experiment, we discarded one additional subject who composed incorrect regexes for 70\% of the tasks (they did not match the example inputs provided in the specification).
Thus, in total, we analyzed the performance of 20 subjects.
We list their demographics in \cref{tab:experiment2_demographics}.


\begin{table}[t]
\centering
\caption{
Demographics of Experiment 2:
  subjects' experience with software development,
  and
  with regexes.
}
\begin{tabular}{lcccccc}
\toprule
~ & \multicolumn{3}{c}{\textbf{Years Prof. Soft. Dev.}} & \multicolumn{3}{c}{\textbf{Exp. with Regexes}} 
\\ 
~ & $<1$ & 1-2 & 3-5 & \emph{Novice} & \emph{Interm.} & \emph{Expert}
\\
\toprule
\textbf{\# Subjects} & 7 & 7 & 6 & 9 & 11 & 0  
\\
\bottomrule
\end{tabular}
\label{tab:experiment2_demographics}
\vspace{-0.1cm}
\end{table}

\mysection{Metrics}
\label{section:Experiment2:metrics}
For each studied task and anti-pattern set, we measured success using \emph{Detection Effectiveness}: the percentage of subjects that correctly identified whether their composed regex was IA.
We used the same approach for ground truth as in Experiment 1 (\cref{section:Experiment1:groundtruth}).
Note that we did not measure whether subjects fixed the IA section in their regex, if any; 
we measured the effectiveness of anti-patterns following the goal of the original SOA anti-patterns --- to identify IA.

\mysection{Statistical Tests}
\label{section:Experiment2:statistical_tests}
We validated our results using hypothesis testing and power analysis.

\mysubsection{Hypothesis testing} 
We used the null hypothesis: \emph{$H_0$: subjects using our IA anti-patterns achieve as much detection effectiveness as those using the SOA anti-patterns.}
We tested $H_0$ using a Wilcoxon signed rank test~\cite{wilcoxon1992individual} (since we cannot assume a normal distribution of results, and our observations are paired) over the IA assessments produced by our IA anti-patterns and the SOA ones for all tasks and orders.
If this test returned a low p-value ($p<0.05$), we rejected $H_0$. 



\mysubsection{Power analysis} 
We used power analysis to determine if our sample size was sufficient to support a statistically significant expected effect size in detection effectiveness~\cite{ellis2010essential}.
We looked for standard power of $0.8$, standard statistical significance of $p<0.05$, with our observed effect size (\ie the difference in detection effectiveness using our IA anti-patterns vs. using the SOA anti-patterns for all tasks and orders).

\begin{table}[t]
\centering
\small
\caption{
    Performance of subjects in Experiment 2:
    percentage of subjects correctly using each anti-pattern set to identify if their composed regex was IA.
}
\setlength\tabcolsep{4pt} 

\resizebox{\linewidth}{!}{  
\begin{tabular}{lcccccc}
\toprule
\multicolumn{1}{c}{\textbf{}}      & \multicolumn{2}{c|}{\textbf{\begin{tabular}[c]{@{}c@{}}SOA first, IA after\\  ($N=10$)\end{tabular}}}                                                                                                        & \multicolumn{2}{c|}{\textbf{\begin{tabular}[c]{@{}c@{}}IA first, SOA after \\  ($N=10$)\end{tabular}}}                                                                                                                     & \multicolumn{2}{c}{\textbf{\begin{tabular}[c]{@{}c@{}}All orders \\ ($N=20$)\end{tabular}}}                                                                                                                                                                                    \\ \cmidrule{2-7} 
\multicolumn{1}{c}{\textbf{Task}} & \multicolumn{1}{c}{\textbf{\begin{tabular}[c]{@{}c@{}}SOA\end{tabular}}} & \multicolumn{1}{c}{\textbf{\begin{tabular}[c]{@{}c@{}}IA\end{tabular}}} & \multicolumn{1}{c}{\textbf{\begin{tabular}[c]{@{}c@{}}SOA\end{tabular}}} & \multicolumn{1}{c}{\textbf{\begin{tabular}[c]{@{}c@{}}IA\end{tabular}}} & \multicolumn{1}{c}{\textbf{\begin{tabular}[c]{@{}c@{}}SOA\end{tabular}}} & \multicolumn{1}{c}{\textbf{\begin{tabular}[c]{@{}c@{}}IA\end{tabular}}} \\ \toprule
\multicolumn{1}{r}{1}             & 100\%                                                                                                                & 100\%                                                                                                 & 100\%                                                                                                                & 100\%                                                                                                 & 100\%                                                                                                                & 100\%                                                                                                \\
\multicolumn{1}{r}{2}             & 10\%                                                                                                                 & 100\%                                                                                                 & 0\%                                                                                                                  & 100\%                                                                                                 & 5\%                                                                                                                  & 100\%                                                                                                \\
\multicolumn{1}{r}{3}             & 20\%                                                                                                                 & 100\%                                                                                                 & 20\%                                                                                                                 & 100\%                                                                                                 & 20\%                                                                                                                 & 100\%                                                                                                \\
\multicolumn{1}{r}{4}             & 30\%                                                                                                                 & 100\%                                                                                                 & 20\%                                                                                                                 & 100\%                                                                                                 & 25\%                                                                                                                 & 100\%                                                                                                  \\
\multicolumn{1}{r}{5}             & 100\%                                                                                                                & 100\%                                                                                                 & 100\%                                                                                                                & 100\%                                                                                                 & 100\%                                                                                                                & 100\%                                                                                                \\ \midrule
\multicolumn{1}{r}{All}     & 52\%                                                                                                                 & 100\%                                                                                                 & 48\%                                                                                                                 & 100\%                                                                                                 & 50\%                                                                                                                 & 100\% \\
\bottomrule
\end{tabular}
}
\vspace{-0.1cm}
\label{tab:user-study}
\end{table}

\subsection{Results} \label{section:Experiment2:results}

\noindent
\cref{tab:user-study} summarizes our results
for each order of application, 
by each set of anti-patterns, 
for each task. 

\begin{itemize}
\item Considering all tasks and treatment orders,
subjects using our IA anti-patterns achieved 100\% detection effectiveness, 
improving on the SOA anti-patterns (50\%).
\item Our hypothesis test showed a statistically significant improvement ($p<.00001$).
\item Our observed effect size was 50\%, comparing the final columns in the bottom row of \cref{tab:user-study}. Our power analysis indicated that we studied a sufficient number of subjects: we needed 11 and studied 20.
\end{itemize}

\noindent
As we expected (\cref{section:Experiment2:tasks}), the SOA anti-patterns showed their limitations in tasks 2, 3, and 4 --- regardless of the ordering.
Also as expected, both sets of anti-patterns achieved 100\% detection effectiveness for tasks 1 and 5, also regardless of ordering.
We conclude that our IA anti-patterns are as effective as the SOA ones when they are not limited, and much more effective than them when they are.





\mysection{Summary for Experiment 2}
Our IA anti-patterns outperformed the SOA anti-patterns when applied manually (100\% vs. 50\% effectiveness).

\section{Experiment 3: Usability when Complementing Existing Tools}
\label{sec:Experiment3}
Experiments 1 and 2 showed that our anti-patterns
  are effective over a wide variety of regexes (\cref{section:Experiment1}),
  and
  can be applied manually by humans (\cref{sec:Experiment2}).
However, developers may prefer automatic tools, \eg for complex regexes.

In Experiment 3, we studied whether our anti-patterns and fix strategies \emph{complement} automatic tools for real-world regexes.
Our goal is not to replace existing automatic tools (we hope developers use them!), but to complement them, to increase developer understanding of the task outcome.

\subsection{Experimental Design} \label{sec:Experiment3-design}


\mysection{Overview} \label{section:Experiment3:process}
We asked 9 software developers to perform real-world \REDOS detection and fixing.
They performed tasks over their own regexes from open-source projects.
They first used only an automatic tool, and then the tool combined with our anti-patterns (for detection) and our fix strategies (for fixing).
Our design was within-subjects.
We fixed the order so we could measure their (hypothesized) increase in understanding after adding our approach.

\mysection{Treatments}  \label{sec:Experiment3:techniques}
For \emph{detection}, our subjects first used a representative detection tool (Weideman \etal's~\cite{Weideman2017StaticExpressions}), and then complemented it with our anti-patterns.
We studied Weideman \etal's approach because they provide a mature implementation with many stars in GitHub.
For \emph{fixing}, our subjects first used a representative fixing approach (van der Merwe \etal's~\cite{VanDerMerwe2017EvilRegexesHarmless}), and then complemented it with our anti-patterns and fix strategies.
We studied van der Merwe \etal's approach because it is the only existing fixing approach that does not modify the language accepted by the regex.\footnote{The algorithm of van der Merwe \etal does not include an open-source implementation. Our implementation is included in our artifact.}

To help these tools to perform their best, we trained our subjects.
For \emph{detection}, we trained them on the purpose and workings of Weideman's detection tool, including the NFA-based characterizations of IA that it detects in the regex (see \cref{fig:ida}).
For \emph{fixing}, we trained them on the purpose and workings of Van Der Merwe's algorithm, \ie that it converts the regex's NFA to an equivalent unambiguous DFA, then back to an equivalent regex.
We also explained our anti-patterns (\cref{tab:our_antipatterns}) and fix strategies (\cref{tab:fix_strategies}). 


\mysection{Tasks} \label{section:Experiment3:tasks}
\label{sec:Experiment3:tasks}
\cref{tab:experiment3_tasks} shows the tasks of Experiment 3.

\mysubsection{Detection Task} 
We asked our subjects to detect vulnerability in 3 regexes,
in a random order: a PDA regex, an EDA regex, and a non-IA regex (see \cref{fig:ida}).
One of the PDA or EDA regexes was the vulnerable one that we took from the subject's software project.
For the other two regexes, we used the same randomly chosen regexes from the dataset studied in Experiment 1 (\cref{section:Experiment1:dataset}).
For each regex, we showed our subjects the output of Weideman's detection tool, and asked them how strongly they understood the vulnerability in the regex.
Then, we also showed them our anti-patterns, asked them to identify the anti-pattern(s) that each regex fits (to prompt them to use the anti-patterns), and asked them the same question again.

\mysubsection{Fixing Task} 
We asked our subjects to fix the vulnerability in the regex that we took from their code.
We showed them their regex in the context of their project, and asked them how strongly they understand it (to refresh their memory).
Then, we showed them the output of van der Merwe's approach: a non-IA version of their regex.
We let them write their own fix or take/adapt van der Merwe's, and we asked them our understanding questions.
Then, we also showed them our anti-patterns and fix strategies, again let them modify their fix if they choose to, and again asked them our understanding questions.

\begin{table}[t]
    \centering
\caption{
    Tasks of Experiment 3.
    Italics denote changing text with each subject.
    Brackets denote subject answers.
}
\begin{tabular}{p{0.13\linewidth}p{0.23\linewidth}p{0.33\linewidth}p{0.1\linewidth} }
        \multicolumn{4}{c}{{\normalsize\textbf{Detection Task}}} \\ \toprule
        ~ & \textbf{Output of automatic detection tool} & \textbf{How strongly do you understand what makes this regex vulnerable?} & \textbf{Explain your reasoning} \\ \midrule
        \emph{PDA regex} & \emph{Output of Weideman's detection tool} \cite{Weideman2017StaticExpressions} & [Very strongly, Strongly, Neutral, Weakly, Very weakly, Not Vulnerable] & [\dots] \\ \midrule
        \emph{EDA regex} & \emph{Output of Weideman's detection tool} \cite{Weideman2017StaticExpressions} & [Very strongly, Strongly, Neutral, Weakly, Very weakly, Not Vulnerable] & [\dots] \\ \midrule
        \emph{Non-IA regex} & \emph{Output of Weideman's detection tool} \cite{Weideman2017StaticExpressions} & [Very strongly, Strongly, Neutral, Weakly, Very weakly, Not Vulnerable] & [\dots] \\ \bottomrule
        \multicolumn{4}{c}{~} \\ 
        \multicolumn{4}{c}{{\normalsize\textbf{Fixing Task}}} \\ \toprule
        ~ & \textbf{Output of automatic fixing tool} & \textbf{How strongly do you understand what makes the resulting fixed regex not vulnerable?} & \textbf{Explain your reasoning} \\ \midrule
        \emph{Their vulnerable regex in context} & \emph{Output of van der Merwe's fixing tool} \cite{VanDerMerwe2017EvilRegexesHarmless} & [Very strongly, Strongly, Neutral, Weakly, Very weakly, Not Vulnerable] & [\dots] \\
        \bottomrule
    \end{tabular}
\label{tab:experiment3_tasks}
\end{table}

\mysection{Within-Subjects Protocol} \label{section:Experiment3:protocol}
Our protocol had three steps:
(1) \emph{Training:}
  This training was more in-depth than Experiment 2 because the subjects had greater expertise.
  We taught the technical details of \REDOS attacks (following~\cite{Davis2018EcosystemREDOS}), and showed the participants how the detection~\cite{Weideman2017StaticExpressions} and fixing~\cite{VanDerMerwe2017EvilRegexesHarmless} tools work.
(2) \emph{Detection:} We asked subjects to detect IA in a set of regexes, first using only existing automatic tools, and then combining them with our anti-patterns.
(3) \emph{Fixing:} We asked subjects to fix an IA regex that they wrote in their codebase to make it non-IA, first using only existing automatic tools, and then combining them with our anti-patterns and fix strategies.

We let subjects ask clarifying questions.
We asked them to think aloud.
To simulate real-world conditions,
we let them use external resources,
and showed them some resources:
 a web interface for the studied tools,
 and
 two websites for regex understanding (\url{www.regex101.com} and \url{www.regexper.com}).
The experiment took $\sim$1 hour per subject.
Subjects were compensated with a \$40 gift card. 

\mysection{Subjects}
\label{section:Experiment3:subjects}
We recruited software developers
that had written a vulnerable regular expression in 
the PyPi~\cite{pypiHomePage} and NPM~\cite{NPMHomePage} software ecosystem.
To increase response rate, we scanned  $\sim$200K PyPi projects and  $\sim$40K NPM projects with some popularity (at least 1 star) and with recent activity (at least one commit since January 2020).
We 
  extracted their regexes;
  discarded any in test files or dependencies;
  and
  identified vulnerable regexes using Davis \etals ensemble of ReDoS detectors~\cite{Davis2019LinguaFranca},
This process resulted in $120$ vulnerable regexes ($99$ from PyPI, $21$ from NPM). 
We
disclosed these potential vulnerabilities to the $120$ software developers who last modified them.
We invited those developers to participate in our experiment. 
9 of them agreed (8\% response rate) --- demographics are in~\cref{tab:experiment3_demographics}.

\begin{table}
\centering
\caption{
Demographics of Experiment 3:
  subjects' experience with software development,
  and
  with regexes.
}
\begin{tabular}{lccccc}
\toprule
~ & \multicolumn{3}{c}{\textbf{Years Prof. Soft. Dev.}} & \multicolumn{2}{c}{\textbf{Exp. with Regexes}} 
\\ 
~ & 3-5 & 6-10 & $>10$ & \emph{Interm.} & \emph{Expert}
\\ \toprule
\textbf{\# Subjects} & 1 & 1 & 7 & 1 & 8  
\\ \bottomrule
\end{tabular}
\label{tab:experiment3_demographics}
\vspace{-.2in}
\end{table}


\mysection{Metrics}
\label{section:Experiment3:metrics}
We measured the success of our anti-patterns or fix strategies as the increase in understanding that our subjects reported after applying them.
We asked our subjects the same question twice, once after applying each treatment.
We also asked them to explain their reasoning (see \cref{tab:experiment3_tasks}).

For \emph{detection}, we asked them how strongly they understood what makes the regex vulnerable, using 
a Likert scale of: \Likert{Very strongly}, \Likert{Strongly}, \Likert{Neutral}, \Likert{Weakly}, and \Likert{Very weakly} understand, and \Likert{Not vulnerable}.
For \emph{fixing}, we asked them how strongly they understood what makes the resulting fixed regex not vulnerable, using the same scale.
Finally, we asked them how helpful they found the anti-patterns or fix strategies for their future overall.

\mysection{Statistical Tests}
\label{section:Experiment3:statistical_tests}
We validated our results using hypothesis testing and power analysis (as in \cref{section:Experiment2:statistical_tests}).

\mysubsection{Hypothesis testing} 
We set two null hypotheses.
For \underline{detection}, 
\emph{$H_0$: Subjects using our IA anti-patterns in combination with existing tools report the same understanding strength of what makes the regex vulnerable as those using existing tools only.}
For \underline{fixing}, 
\emph{$H_0$: Subjects using our fix strategies in combination with existing tools report the same understanding strength of what makes the fixed regex not vulnerable as those using existing tools only.}
We tested the null hypothesis for each task using a Wilcoxon signed rank test~\cite{wilcoxon1992individual} (since we cannot assume a normal distribution of results, and our observations are paired) over the reported understanding scores for each treatment.

\mysubsection{Power analysis} 
We 
again looked for standard power of $0.8$, standard statistical significance of $p<0.05$, and
measured effect size as the increase in mean reported understanding for each task.

\subsection{Results} \label{section:Experiment3:resulst}

\vspace{-5px}
\begin{figure}[] 
	\centering
    \includegraphics[width=1\columnwidth]{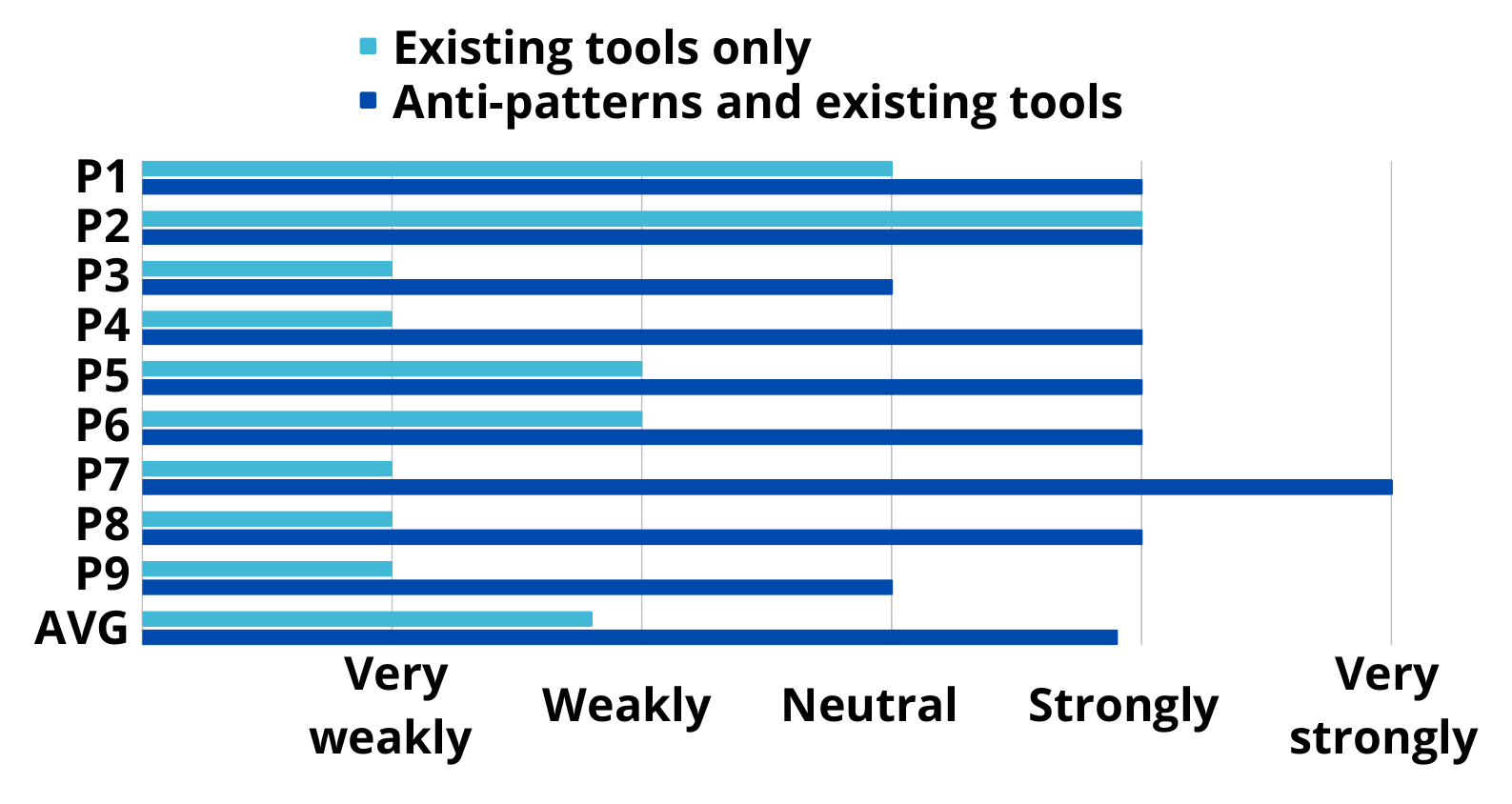}
    \captionof{figure}{
        Detection Task: Subjects consistently reported stronger understanding of what makes their regex vulnerable when using our anti-patterns to complement existing tools.
    }
    \label{fig:result_task 2_exp3}
    \vspace{-15px}
\end{figure}


\mysection{Detection Task}
We focus on how strongly subjects understood what makes their own regex vulnerable (see \cref{fig:result_task 2_exp3}).
\begin{itemize}
	\item Subjects using our anti-patterns to complement existing tools reported median \Likert{Strongly} understanding the vulnerability, improving over using existing tools only (median \Likert{Very weakly}).
	\item Our hypothesis test showed a statistically significant improvement ($p<0.05$).
	\item Our observed effect size was mean 2.1 Likert points --- bottom bar in \cref{fig:result_task 2_exp3}.
	Our power analysis indicated that we studied a sufficient number of subjects: we needed 4 and studied 9.
\end{itemize}

\noindent
Subjects also reported that the anti-patterns will be \Likert{Helpful} ($N=4$) or \Likert{Very helpful} ($N=5$) for their future detection efforts.
The following quote describes their most common sentiment:
\myinlinequote{I will use the tool to see if there is something wrong, and with the anti-patterns I can try to understand why there is a problem}.




\begin{figure}[h] 
	\centering
    \includegraphics[width=1\columnwidth]{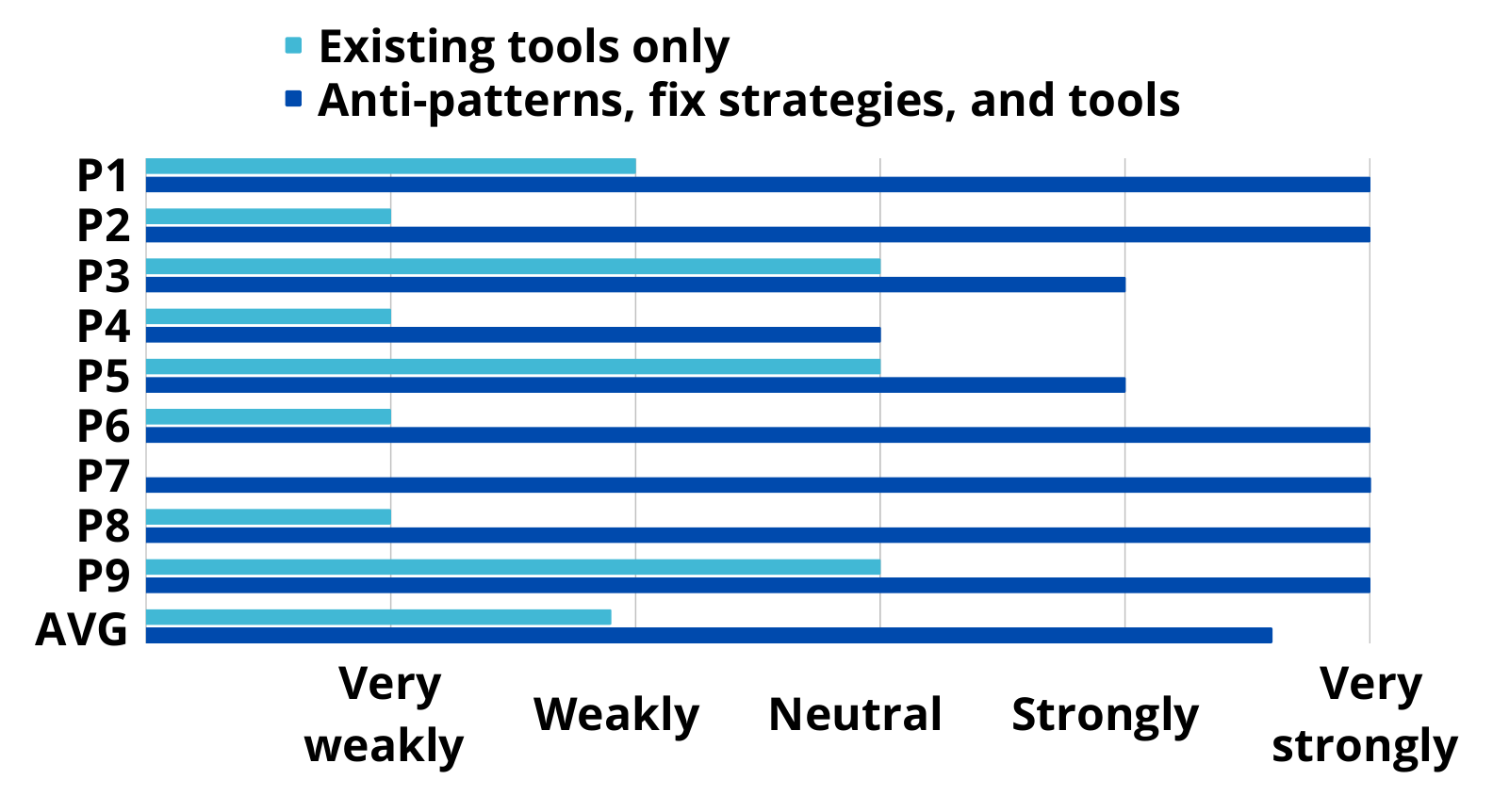}
    \captionof{figure}{
        Fixing Task: Subjects consistently reported stronger understanding of what makes their resulting fixed regex not vulnerable when they used our anti-patterns and fix strategies to complement existing tools.
    }
    \label{fig:result_repair_exp3}
\end{figure}

\mysection{Fixing Task}
\cref{fig:result_repair_exp3} shows our subjects' reported understanding of what makes the fixed regex not vulnerable. 
\begin{itemize}
	\item Subjects using our anti-patterns and fix strategies to complement existing tools reported median \Likert{Very strongly} understanding, improving over using existing tools only (median \Likert{Very weakly}).
	\item Our hypothesis test showed a statistically significant improvement ($p<0.05$).
	\item Our observed effect size was mean 2.9 Likert points --- bottom bar in \cref{fig:result_repair_exp3}.
	Our power analysis indicated that we studied a sufficient number of subjects: we needed 2 and studied 9.
\end{itemize}

\noindent
Subjects also reported that the anti-patterns and fix strategies will be \Likert{Neutral} ($N=1$), \Likert{Helpful} ($N=2$), or \Likert{Very helpful} ($N=6$) for their future detection efforts.
As an example quote, one subject regarded the fixing tool as:
    \myinlinequote{The output does not make a whole lot of sense to me}.
Another said of the fix resulting from our fix strategies:
    \myinlinequote{I understand why this is ambiguous and how the change fixes it}.
Finally, almost all subjects ($N=8$) were more comfortable fixing their codebase with the fix produced using our fix strategies than with the one produced by the existing tool (\cref{fig:result_repair2_exp3} in Appendix).


\mysection{Summary for Experiment 3}
Subjects using our IA anti-patterns and fix strategies to complement existing tools reported much higher understanding, from median \Likert{Very weakly} to median \Likert{Strongly} for detection, and to median \Likert{Very strongly} for fixing.

\section{Threats to Validity} \label{section:threats}

\mysection{Internal Validity}
We took multiple measures to increase internal validity.
In Experiment 1 (\cref{section:Experiment1}), 
we tested the implementation scripts of our anti-patterns over small samples of the dataset.
We also curated our studied dataset to prepare it for our experiments.
We also used existing implementations of tools where possible, viz. the SOA anti-patterns by Davis \etal~\cite{Davis2018EcosystemREDOS} and Weidemann \etal's detector~\cite{Weideman2017StaticExpressions}, to avoid errors if implementing them ourselves.

In Experiment 2 (\cref{sec:Experiment2}),
we used best practices in human-subject experiment methodologies in its design, \eg~\cite{kitchenham2008personal,siegmund2014measuring}.
We 
\emph{piloted} the protocol on
3 pilot studies, which helped us clarify
the language describing the tasks (pilot 1) and the technical terms in the training (pilot 2).
Pilot 3 showed that our script was adequate.
To reduce \emph{social desirability bias}, we did not disclose who created any of the anti-patterns, and we referred to them in the third person (Anti-patterns 1 and 2).
To avoid \emph{expertise bias}, we asked subjects to apply the given anti-patterns irrespective of their perception of their correctness.
To avoid \emph{learning bias}, 10 random subjects used the SOA anti-patterns first, and the other 10 used our IA anti-patterns first.

In Experiment 3 (\cref{sec:Experiment3})
we likewise used best practices in design.
We \emph{piloted} the protocol on 3 subjects.
After pilot 1, we adjusted the number of tasks to reduce the experiment duration.
After pilots 2 and 3, we clarified some terms in the training.
To reduce \emph{social desirability bias}, we did not disclose who created any of the tools or supplementary approaches (even if we were asked); we referred to them only as ``Treatment 1'' etc.
Further, we separately and independently asked subjects their understanding using one treatment and then using the other (as opposed to asking them to compare treatments).
Having subjects individually decide and assign a specific score to each treatment reduces the possibility of them unconsciously preferring the last treatment.
To reduce \emph{expertise bias} (\eg higher understanding reported by more experienced subjects), all subjects used both treatments.
To reduce \emph{learning bias}, subjects used first the treatment that we anticipated would provide lower understanding, \ie the existing tools.
If subjects used first the treatment that truly provided higher understanding (and second the one that truly produced lower), they would misleadingly report higher understanding for the second treatment; they cannot forget what they learned.
Having the combination of our anti-patterns, fix strategies, and existing tools as the last treatment may have unconsciously nudged subjects to report higher understanding for them.
However, subjects reported much higher understanding for this last treatment with statistical significance, \ie more likely due to a real effect than to chance.

\mysection{External Validity}
We also took multiple measures to increase external validity.

In Experiment 1 (\cref{section:Experiment1}), 
we evaluated the largest available dataset of regexes~\cite{Davis2019LinguaFranca}.

In Experiment 2 (\cref{sec:Experiment2}),
our subjects had diverse levels of professional software development experience.
Their experience with regexes was at the novice and intermediate levels, but we studied regex experts in Experiment 3.
We studied simple composition tasks to represent situations when developers may choose to apply anti-patterns manually, but they represented diverse scenarios.
Furthermore, our subjects composed solutions that were only slightly less complex than typical real-world regexes according to \cite{Davis2019RegexGeneralizability}. 
For example, they had length 6-11 (median regex length in Java: 15) and used 2-3 operators (Java: median 3).
We report the most common solution observed for each task in \cref{tab:Regex Composition Tasks}.
We also studied more complex real-world regexes in Experiment 3.

In Experiment 3 (\cref{sec:Experiment3}),
we studied mostly regex experts and complex real-world regexes to complement Experiment 2.
We also made this experiment as realistic as possible by having developers work with their own regexes, and giving them free access to online resources.

Finally, both Experiment 2 and 3 studied a limited number of subjects ($N=20$ and $N=9$).
However, we observed large effect sizes in our results, we did so consistently, they were statistically significant, and power analysis revealed that fewer subjects would have been sufficient.

\section{Conclusions} \label{section:conclusions}

To secure software systems, developers need approaches that are both sound and understandable.
Prior to this paper, the approaches to address regular expression security problems provided theoretical guarantees, but were difficult for developers to understand.
Our goal was to complement these existing approaches with understandable regex security anti-patterns and fix strategies.
To that end, we developed a novel theory of regex infinite ambiguity that characterizes vulnerable regexes to \REDOS, and a set of anti-patterns and fix strategies derived from it.
Our evaluation showed that our IA anti-patterns identified vulnerable regexes with much higher effectiveness than the state-of-the-art anti-patterns, both when applied automatically and manually.
Our anti-patterns and fix strategies also substantially increased developer understanding when used alongside existing tools to detect and fix vulnerable regexes.
In the future, we plan to apply this methodology to similar security problems in domain-specific languages (\eg in GraphQL~\cite{cha2020principled}).

\vspace{-0.1cm}
\section*{Research Ethics}
\vspace{-0.1cm}

Our human subjects experiments were overseen by the appropriate Institutional Review Board (IRB).

\vspace{-0.1cm}
\section*{Acknowledgments}
\vspace{-0.1cm}

We thank the reviewers for their constructive feedback.
We thank Charles M. Sale for developing the \url{www.regextools.io} platform for our experiment.
Lee and Davis acknowledge support from NSF award \#2135156,
and Servant from URJC award C01INVESDIST.

\raggedbottom
\pagebreak

\begin{appendices}

\section{Replication Package}
\label{section:Appendix-replication}

We have made our data and code publicly available for replication~\cite{hassan2022}.

It contains, for Experiment 1:
\emph{i)} the dataset used (\cref{section:Experiment1:dataset}), 
\emph{ii)} our implementation of Weideman's detection tool~\cite{Weideman2017StaticExpressions} used as ground truth (\cref{section:Experiment1:groundtruth}), and
\emph{iii)} the implementation of our anti-patterns (\cref{section:Experiment1:techniques}).
For Experiments 2 and 3:
\emph{iv)} the full protocol used (\cref{section:Experiment2:protocol} and \cref{section:Experiment3:protocol}), and
\emph{v)} our implementation of van der Merwe's fixing tool~\cite{VanDerMerwe2017EvilRegexesHarmless} (\cref{sec:Experiment3:techniques}).
Finally,
\emph{vi)} the analysis of prevalence of our fix strategies (\cref{sec:fix_strategies:popularity}).

\section{Proofs of the Theorems} \label{section:Appendix-Proofs}

\subsection{Definitions}
We define the operators used in Theorems 2 and 3:

\subsubsection{$\xblnot$}
Brabrand \& Thomsen~\cite{brabrand2010typed} introduced an overlap operator, $\xblnot$, between two languages $L(R_1)$ and $L(R_2)$.
The set $L(R_1) \ \xblnot \ L(R_2)$ contains the ambiguity-inducing strings that can be 
parsed in multiple ways across $L(R_1)$ and $L(R_2)$.
More formally, with $X = L(R_1)$ and $Y = L(R_2)$,
\[
X \ \xblnot \ Y = \{ xay \ | \ x,y \in {\Sigma}^{*} \wedge a \in {\Sigma}^{+} s.t. \ x,xa \in X \wedge ay,y \in Y \}
\]

\subsubsection{$\Omega$}

Brabrand \& Thomsen use M{\o}ller's BRICS library~\cite{moller2010dk} for the implementation of their theorems, and actually use what we call the ``M{\o}ller overlap operator'', $\Omega$.
We use this operator in our theorems.
The M{\o}ller overlap operator describes only the ambiguous core ``$a$'':

\[
X \ \Omega \ Y = \{ \exists \ x, y \in {\Sigma}^{*} \wedge a \ | \ a \in {\Sigma}^{+}  s.t. \ x, xa \in X \wedge ay, y \in Y \}
\]

\subsection{Assumptions}
In our theorems and proofs, we assume that we can convert regexes to their equivalent, ambiguity-preserving, $\epsilon$-free NFAs~\cite{weber1991degree,Weideman2017StaticExpressions}.
\subsection{Theorems \& proofs}
Brabrand \& Thomsen's~\cref{thm:Brabrand}~\cite{brabrand2010typed} provides the conditions for \emph{unambiguity}.
Our proofs consider the effect of negating the unambiguity condition, and distinguish the conditions that lead to finite or infinite ambiguity.

\subsubsection{Theorem 0: Brabrand \& Thomsen's~\cite{brabrand2010typed} Theorem}
\textit{Given unambiguous regexes $R_1$ and $R_2$,}
    \begin{enumerate}[label=(\alph*)]
        \item $R_1|R_2$ is unambiguous iff $L(R_1) \cap L(R_2) = \phi$.
        \item $R_1{\cdot}R_2$ is unambiguous iff $L(R_1) \ \xblnot \ L(R_2) = \phi$.
        \item $R_{1}*$ is unambiguous iff $\epsilon \notin L(R_1) \wedge L(R_1) \ \xblnot \ L(R_{1}*) = \phi$.
    \end{enumerate}
\subsubsection{Theorem 1: Ambiguity of Alternation}
\textit{Given unambiguous regexes $R_1$ and $R_2$,}
    \begin{enumerate}[label=(\alph*)]
        \item \textit{$R_1|R_2$ is finitely ambiguous iff $L(R_1) \cap L(R_2) \neq \phi$}.
        \item \textit{$R_1|R_2$ cannot be infinitely ambiguous.}
    \end{enumerate}

The components of~\cref{thm:alternation} follow from~\cref{lemma:alternation:A:FA}.

\begin{lemma}
\label{lemma:alternation:A:FA}
Given unambiguous $R_1$ and $R_2$, if $R_1|R_2$ is ambiguous it is always finitely ambiguous.
\end{lemma}
\begin{proof}
A string $s$ may be matched against $R_1|R_2$ in four ways:
  $s$ may be matched by $R_1$, by $R_2$, by both, or by neither.
In any case, since $R_1$ and $R_2$ are unambiguous, there are at most two ways for $R_1|R_2$ to match $s$.

\end{proof}

\subsubsection{Theorem 2: Ambiguity of Concatenation}

\textit{Suppose unambiguous regexes $R_1$ and $R_2$,
    and that
      $L(R_1) \ \xblnot \ L(R_2) \neq \phi$
    (so $R_1{\cdot}R_2$ is ambiguous by~\cref{thm:Brabrand}).
    Then:}
    \begin{enumerate}[label=(\alph*)]
        \item $R_1{\cdot}R_2$ is infinitely ambiguous iff 
        $L(R_1)$ contains the language of a regex \verb<BC*D< and 
        $L(R_2)$ contains the language of a regex \verb<EF*G<, 
        where $\epsilon \notin L($\verb<C<$) \wedge \epsilon \notin L($\verb<F<$) \wedge L($\verb<C<$) \cap L($\verb<F<$) \cap L($\verb<DE<$) \neq \phi$.
        \item Otherwise, $R_1{\cdot}R_2$ must be finitely ambiguous.
    \end{enumerate}

2(a) is an iff so we need to prove:

$\impliedby$: 
If 
$L(R_1)$ contains the language of a regex \verb<BC*D< and 
$L(R_2)$ contains the language of a regex \verb<EF*G<, 
where 
$\epsilon \notin L($\verb<C<$) \wedge 
\epsilon \notin L($\verb<F<$) \wedge
L($\verb<C<$) \cap L($\verb<F<$) \cap L($\verb<DE<$) \neq \phi$,
then
$R_1\cdot R_2$ is infinitely ambiguous.

\begin{proof}

Consider a string $q=bc^md \in L($\verb<BC*D<$)$ where $b,d \in \Sigma^\ast$, $c \in \Sigma^{+}$, $b \in L($\verb<B<$)$, $c \in L($\verb<C<$)$, and $d \in L($\verb<D<$)$.
By hypothesis, $L($\verb<BC*D<$) \subseteq L(R_1)$, so $q \in L(R_1)$.
Similarly, consider another string $r=ef^ng \in L($\verb<EF*G<$)$ where $e,g \in \Sigma^\ast$, $f \in \Sigma^{+}$, $e \in L($\verb<E<$)$, $f \in L($\verb<F<$)$, and $g \in L($\verb<G<$)$.
By hypothesis, $L($\verb<EF*G<$) \subseteq L(R_2)$, so $r \in L(R_2)$.
As $L($\verb<C<$) \cap L($\verb<F<$) \cap L($\verb<DE<$) \neq \phi$, suppose $c=f=de$.

Consider the new string $p = qr = bc^mdef^ng \in L(R_1){\cdot}L(R_2) = L(R_1{\cdot}R_2)$.
In other words,
$R_1{\cdot}R_2$ should include the following NFA accepting $p$.

\begin{figure}[ht] 
\centering 
\begin{minipage}[b]{\columnwidth}
  \centering
  \begin{tikzpicture}[baseline=-1.80em,initial text=]
    \tikzstyle{every state}=[inner sep=1pt, minimum size=0.2cm]
    
    \node[state, initial] (q0) {$...$};
    \node[state, right = 0.4cm of q0] (q1) {$v_1$}; 
    \node[state, right = 0.4cm of q1] (q2) {$v_2$}; 
    \node[state, right = 0.4cm of q2] (q3) {$v_3$}; 
    \node[state, right = 0.4cm of q3] (q4) {$...$}; 

    \draw [->] (q0) edge node [above] {\tt $b$} (q1);
    \draw [->] (q1) edge node [above] {\tt $d$} (q2);
    \draw [->] (q2) edge node [above] {\tt $e$} (q3);
    \draw [->] (q3) edge node [above] {\tt $g$} (q4);
    
    \draw (q1) edge[loop above] node [left] {\tt $c$} (q1);
    \draw (q3) edge[loop above] node [left] {\tt $f$} (q3);
    
  \end{tikzpicture} 
\end{minipage}
\end{figure}

For $m=2$ and $n=2$, $p=bccdeffg$. 
There are $(m\times n)+1 =(2\times2)+1 = 5$ ways to match. Ignoring prefix $b$ and suffix $g$, the five cases to match the middle $ccdeff$ are:
\begin{itemize}

\item
$v_1 \rightarrow^c v_1 \rightarrow^c v_1 \rightarrow^{(de=c)} v_1 \rightarrow^{(f=c)} v_1 \rightarrow^{(f=de)} v_3$ 

\item
$v_1 \rightarrow^c v_1 \rightarrow^c v_1 \rightarrow^{(de=c)} v_1 \rightarrow^{(f=de)} v_3 \rightarrow^f v_3$ 

\item
$v_1 \rightarrow^c v_1 \rightarrow^c v_1 \rightarrow^d v_2 \rightarrow^e v_3 \rightarrow^f v_3 \rightarrow^f v_3$ 

\item
$v_1 \rightarrow^c v_1 \rightarrow^{(c=de)} v_3 \rightarrow^{(de=f)} v_3 \rightarrow^f v_3 \rightarrow^f v_3$ 

\item
$v_1 \rightarrow^{(c=de)} v_3 \rightarrow^{c=f} v_3 \rightarrow^{(de=f)} v_3 \rightarrow^f v_3 \rightarrow^f v_3$ 

\end{itemize}
\noindent
where the superscript of an arrow represents the (input observed = path taken) pair.


The degree of ambiguity grows for each larger $m$ and $n$.
It can be shown that for an input string $p = bc^mdef^ng$, there will be $(m\times n)+1$ ways to match. 
Here ambiguity is a function of the input length.
Therefore, $R_1{\cdot}R_2$ is infinitely ambiguous. 


\end{proof}

$\implies$: 
If $R_1\cdot R_2$ is infinitely ambiguous, then
$L(R_1)$ contains the language of a regex \verb<BC*D< and 
$L(R_2)$ contains the language of a regex \verb<EF*G<, 
where $\epsilon \notin L($\verb<C<$) \wedge \epsilon \notin L($\verb<F<$) \wedge L($\verb<C<$) \cap L($\verb<F<$) \cap L($\verb<DE<$) \neq \phi$.
\begin{proof}

We will reason over an equivalent, ambiguity-preserving, $\epsilon$-free NFA~\cite{weber1991degree}.
The NFA of an infinitely ambiguous regex should include either a Polynomial or an Exponential Degree of Ambiguity (PDA, EDA) section~\cite{weber1991degree}, as shown in~\cref{fig:ida}.

We first show that 
if $R_1\cdot R_2$ is infinitely ambiguous, then
the NFA of $R_1\cdot R_2$ must contain a PDA (\cref{fig:ida}(a)).
$R_1$ and $R_2$ are unambiguous, so none of them should have a full EDA.
Concatenating two regexes $R_1\cdot R_2$ cannot create a new self loop of EDA.
Thus, $R_1\cdot R_2$ must contain a PDA.

Consider the two nodes $p$ with the loop $\pi_1$ and $q$ with the loop $\pi_3$ in \cref{fig:ida}(a).
As $R_1$ and $R_2$ are unambiguous, neither $R_1$ nor $R_2$ can include both nodes $p$ and $q$ --- because then they would be infinitely ambiguous (not unambiguous).
Therefore, $R_1$ and $R_2$ each should have a part of PDA; and the partition will appear somewhere along the path $\pi_2$ as the loops $\pi_1$ and $\pi_3$ cannot be newly introduced via concatenation.

Each partition of PDA consists of a prefix, a loop, and a suffix,
which can be mapped to a regex of the form \verb>PQ*R>. 
As a PDA is a part of the whole NFA, more generally, we can conclude that
(1) $L(R_1)$ contains the language of a regex \verb<BC*D< and 
(2) $L(R_2)$ contains the language of a regex \verb<EF*G<:
\ie $L($\verb<BC*D<$) \subseteq L(R_1)$ and $L($\verb<EF*G<$) \subseteq L(R_2)$ where $\epsilon \notin L($\verb<C<$) \wedge \epsilon \notin L($\verb<F<$)$.


After concatenation, the full PDA can be represented by a language of the form \verb<BC*DEF*G<,
where \verb<C*< is mapped to the first loop $\pi_1$, \verb<DE< to the path $\pi_2$, and \verb<F*< to the second loop $\pi_3$.
Let $s$ be the string that meets the PDA path conditions: 
$label(\pi_1)=label(\pi_2)=label(\pi_3)$.
Then, $s \in L($\verb<C<$)$ (by $label(\pi_1$)), $s \in L($\verb<DE<$)$, and $s \in L($\verb<F<$)$.
And thus 
$L($\verb<C<$) \cap L($\verb<F<$) \cap L($\verb<DE<$) \neq \phi$.

\end{proof}

\cref{thm:concat}(b) follows from elimination with \cref{thm:Brabrand}.

\subsubsection{Theorem 3: Ambiguity of Star}
\textit{Given unambiguous regex R,} 
    \begin{enumerate}[label=(\alph*)]
        \item \textit{$R*$ is infinitely ambiguous iff $ \epsilon \in L(R) \vee L(R) \ \Omega \ L(R*) \neq \phi$}.
        \item \textit{$R*$ cannot be finitely ambiguous.}
    \end{enumerate}

The components of~\cref{thm:star} follow from~\cref{lemma:star:A:IA}.

\begin{lemma}
\label{lemma:star:A:IA}
If R* is ambiguous, it is always infinitely ambiguous.
\end{lemma}
\begin{proof}


We prove this by induction.
From the contrapositive of~\cref{thm:Brabrand}(c),
if $R*$ is ambiguous, $L(R) \ \xblnot \ L(R*) \neq \phi$.
There exists an input string $s=xay$ such that 1) $x,y  \in \Sigma^\ast$, 2) $a \in \Sigma^+$, 3) $x, xa \in L(R)$, 4) $y,ay \in L(R\ast)$.
In other words, there are at least $2=2^1$ ways to parse $s$ (\ie $x \in L(R)$ then $ay \in L(R\ast)$; or $xa$ then $y$). 

Now consider $ss = (xay)(xay)$. Let $x'=x, a'=a, y'=yxay$ then,
$ss=x'a'y'$. Then the following conditions are true:
  (1) $x', y' \in \Sigma^\ast$,
  (2) $a' \in\Sigma^\ast$,
  (3) $x',x'a' \in L(R)$,
  and
  (4) $y', a'y' \in L(R*RR*) \subset L(R*)$.
  For each $xay$ there are at least 2 accepting paths.
  Therefore, for $ss$ there are at least $4$=$2^2$ accepting paths, and the degree of ambiguity grows for each additional concatenation of an $s$. 
Therefore, $R\ast$ is infinitely ambiguous. 

\end{proof}

\subsubsection{Theorem 4: Finite to Infinite}
\textit{Given a finitely ambiguous regex $R$, $R\ast$ is always infinitely ambiguous.}
\begin{proof}
If $R$ is finitely ambiguous by definition there exists an input string $s$ for which there will be at least 2 accepting paths. For $R\ast$, we can increase the length of input string as much as we want because of the $\ast$. Now for input string $ss$, there will be at least $4=2^2$ accepting paths as we have at least 2 options for each $s$. By the same logic, for input string $sss....$ where length of the input string is $n$, there will be at least $2^n$ accepting paths.

Therefore, $R\ast$ is infinitely ambiguous.
\end{proof}

\subsection{Limitations}
Our theorems do not cover the cases when $R1$ and $R2$ are finitely ambiguous.
In such scenario, our expectation is that Alternation ($R_1|R_2$) would always yield a finitely ambiguous regex.
We also expect that Concatenation ($R_1\cdot R_2$) would still yield an infinitely ambiguous regex if 
$L(R_1)$ contains the language of a regex \verb<BC*D< and 
$L(R_2)$ contains the language of a regex \verb<EF*G<, 
where $\epsilon \notin L($\verb<C<$) \wedge \epsilon \notin L($\verb<F<$) \wedge L($\verb<C<$) \cap L($\verb<F<$) \cap L($\verb<DE<$) \neq \phi$.
However, whether this is the only case is less clear.
Still, despite this limitation, our theorems allowed us to derive anti-patterns (\cref{section:AntiPatterns}) and fix strategies (\cref{sec:fix_strategies}) that substantially improved the effectiveness of the SOA ones (\cref{section:Experiment1}, \cref{sec:Experiment2}), and the usability of existing automatic detection and fixing tools (\cref{sec:Experiment3}).

\section{Other Figures}
\label{Appendix-OtherFigs}
\subsection{CVEs Increasing Year by Year}\label{CVEs}

We observe annual growth in \REDOS CVEs from 2010 to the present.
\cref{fig:CVES} shows the trend of \REDOS CVEs since 2010.
The incidence of \REDOS CVEs grew from 2 in 2010 to 20 in 2021.

\begin{figure}[h] 
	\centering
    \includegraphics[width=0.9\columnwidth]{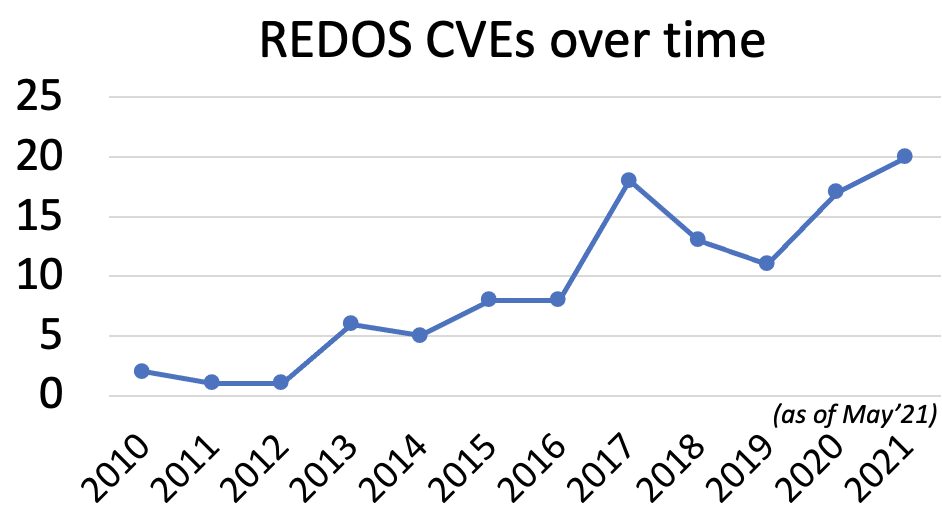}
    \captionof{figure}{
        The data were obtained by a two-step process: a preliminary labeling of the CVE database using key words and phrases (\eg ``\REDOS'' or ``extremely long time'' with a reference to regular expressions), followed by a manual inspection for accuracy.
    }
    \label{fig:CVES}
\end{figure}

\subsection{Fix Acceptance}\label{exp3_repair_comfort}

We asked the participants of Experiment 3 (\cref{sec:Experiment3}) how comfortable they were replacing the vulnerable regex in their codebase with the fixes provided by each repair treatment.
As shown in~\cref{fig:result_repair2_exp3}, almost all subjects were more comfortable with the fix produced using our anti-patterns and fix strategies. 

\begin{figure}[h] 
	\centering
    \includegraphics[width=1\columnwidth]{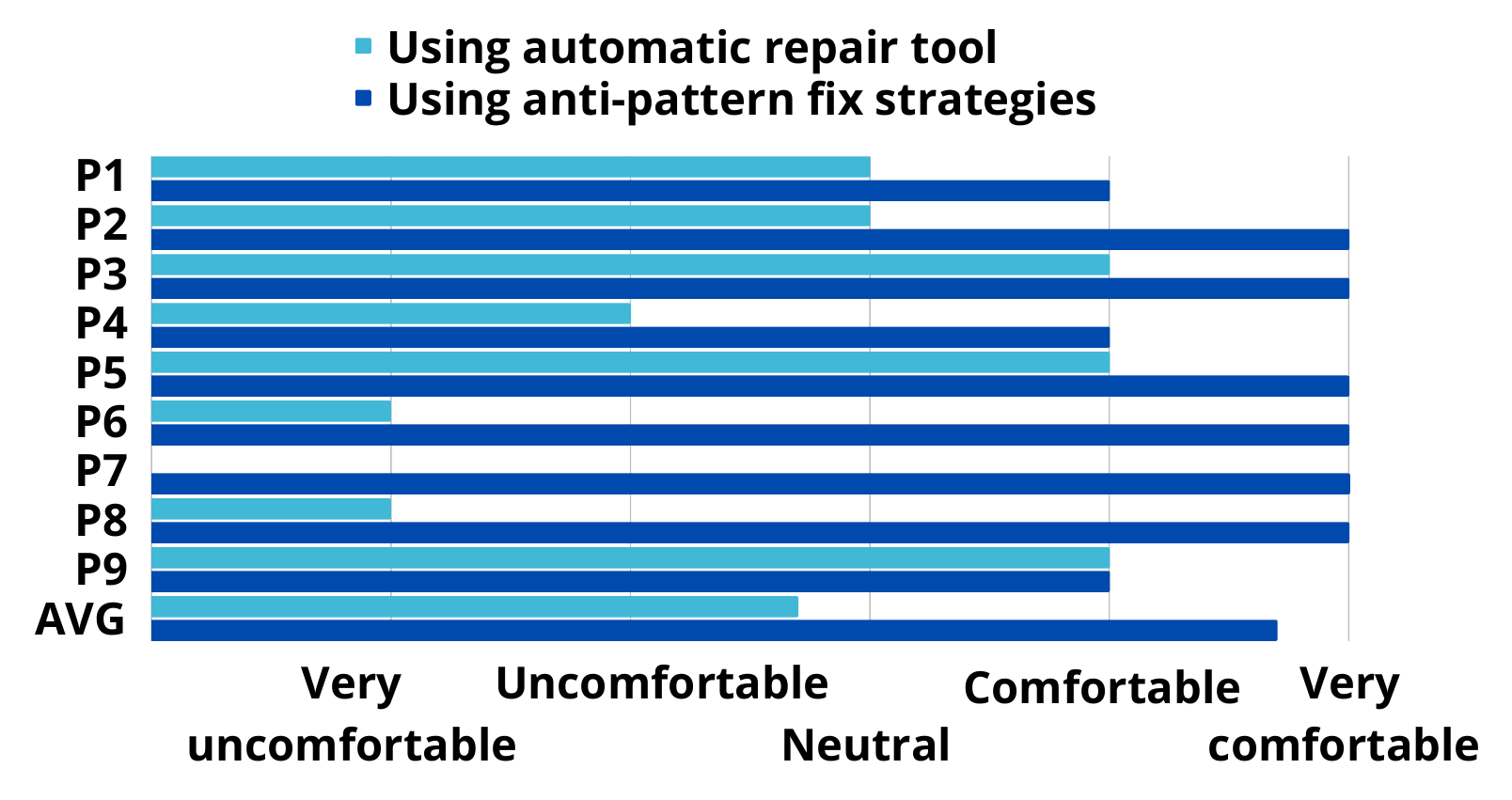}
    \captionof{figure}{
        Fixing Task: How comfortable our subjects reported being with fixing their codebase with the fix produced by each treatment.
    }
    \label{fig:result_repair2_exp3}
\end{figure}
\end{appendices}


\raggedbottom
\newpage
\balance

\bibliographystyle{IEEEtran}
\bibliography{bib/Mendeley, bib/WebLinks, bib/misc}

\begin{thebibliography}{10}
\providecommand{\url}[1]{#1}
\csname url@samestyle\endcsname
\providecommand{\newblock}{\relax}
\providecommand{\bibinfo}[2]{#2}
\providecommand{\BIBentrySTDinterwordspacing}{\spaceskip=0pt\relax}
\providecommand{\BIBentryALTinterwordstretchfactor}{4}
\providecommand{\BIBentryALTinterwordspacing}{\spaceskip=\fontdimen2\font plus
\BIBentryALTinterwordstretchfactor\fontdimen3\font minus
  \fontdimen4\font\relax}
\providecommand{\BIBforeignlanguage}[2]{{%
\expandafter\ifx\csname l@#1\endcsname\relax
\typeout{** WARNING: IEEEtran.bst: No hyphenation pattern has been}%
\typeout{** loaded for the language `#1'. Using the pattern for}%
\typeout{** the default language instead.}%
\else
\language=\csname l@#1\endcsname
\fi
#2}}
\providecommand{\BIBdecl}{\relax}
\BIBdecl

\bibitem{Chapman2016RegexUsageInPythonApps}
C.~Chapman and K.~T. Stolee, ``{Exploring regular expression usage and context
  in Python},'' \emph{International Symposium on Software Testing and Analysis
  (ISSTA)}, 2016.

\bibitem{Davis2018EcosystemREDOS}
J.~C. Davis, C.~A. Coghlan, F.~Servant, and D.~Lee, ``{The Impact of Regular
  Expression Denial of Service (ReDoS) in Practice: an Empirical Study at the
  Ecosystem Scale},'' in \emph{The ACM Joint European Software Engineering
  Conference and Symposium on the Foundations of Software Engineering
  (ESEC/FSE)}, 2018.

\bibitem{Li2008REL}
\BIBentryALTinterwordspacing
Y.~Li, R.~Krishnamurthy, S.~Raghavan, S.~Vaithyanathan, and H.~V. Jagadish,
  ``Regular expression learning for information extraction,'' in
  \emph{Proceedings of the Conference on Empirical Methods in Natural Language
  Processing (EMNLP 08)}, 2008, pp. 21--30. [Online]. Available:
  \url{http://dl.acm.org/citation.cfm?id=1613715.1613719}
\BIBentrySTDinterwordspacing

\bibitem{Gogte2016HARE}
V.~Gogte, A.~Kolli, M.~J. Cafarella, L.~D'Antoni, and T.~F. Wenisch, ``Hare:
  Hardware accelerator for regular expressions,'' in \emph{2016 49th Annual
  IEEE/ACM International Symposium on Microarchitecture (MICRO)}, Oct 2016, pp.
  1--12.

\bibitem{Chiticariu2011SystemT}
\BIBentryALTinterwordspacing
L.~Chiticariu, V.~Chu, S.~Dasgupta, T.~W. Goetz, H.~Ho, R.~Krishnamurthy,
  A.~Lang, Y.~Li, B.~Liu, S.~Raghavan, F.~R. Reiss, S.~Vaithyanathan, and
  H.~Zhu, ``The systemt ide: An integrated development environment for
  information extraction rules,'' in \emph{Proceedings of the 2011 ACM SIGMOD
  International Conference on Management of Data (SIGMOD 11)}, 2011, pp.
  1291--1294. [Online]. Available:
  \url{http://doi.acm.org/10.1145/1989323.1989479}
\BIBentrySTDinterwordspacing

\bibitem{efftinge2006oaw}
S.~Efftinge and M.~V{\"o}lter, ``oaw xtext: A framework for textual dsls,'' in
  \emph{Workshop on Modeling Symposium at Eclipse Summit}, vol.~32, 2006, p.
  118.

\bibitem{Balzarotti2008Saner}
D.~Balzarotti, M.~Cova, V.~Felmetsger, N.~Jovanovic, E.~Kirda, C.~Kruegel, and
  G.~Vigna, ``{Saner: Composing static and dynamic analysis to validate
  sanitization in web applications},'' in \emph{IEEE Symposium on Security and
  Privacy (IEEE S{\&}P)}, 2008, pp. 387--401.

\bibitem{Wassermann2008XSS}
\BIBentryALTinterwordspacing
G.~Wassermann and Z.~Su, ``Static detection of cross-site scripting
  vulnerabilities,'' in \emph{Proceedings of the 30th International Conference
  on Software Engineering (ICSE '08)}, 2008, p. 171–180. [Online]. Available:
  \url{https://doi.org/10.1145/1368088.1368112}
\BIBentrySTDinterwordspacing

\bibitem{ModSecurityWAFRuleset}
``Owasp modsecurity core rule set,,'' \url{https://coreruleset.org/}.

\bibitem{ClamAVRegexRules}
N.~L. Or, X.~Wang, and D.~Pao, ``Memory-based hardware architectures to detect
  clamav virus signatures with restricted regular expression features,''
  \emph{IEEE Transactions on Computers}, vol.~65, no.~4, pp. 1225--1238, 2016.

\bibitem{Crosby2003REDOS}
S.~Crosby, ``{Denial of service through regular expressions},'' \emph{USENIX
  Security work in progress report}, 2003.

\bibitem{Roichman2009ReDoS}
A.~Roichman and A.~Weidman, ``{VAC - ReDoS: Regular Expression Denial Of
  Service},'' \emph{Open Web Application Security Project (OWASP)}, 2009.

\bibitem{2016StackOverflowOutage}
S.~Exchange, ``Outage postmortem,''
  \url{http://web.archive.org/web/20180801005940/http://stackstatus.net/post/147710624694/outage-postmortem-july-20-2016},
  2016.

\bibitem{Cloudflare2019REDOSPostMortem}
{Graham-Cumming, John}, ``Details of the cloudflare outage on july 2, 2019,''
  \url{https://web.archive.org/web/20190712160002/https://blog.cloudflare.com/details-of-the-cloudflare-outage-on-july-2-2019/}.

\bibitem{Staicu2018REDOS}
\BIBentryALTinterwordspacing
C.-A. Staicu and M.~Pradel, ``{Freezing the Web: A Study of ReDoS
  Vulnerabilities in JavaScript-based Web Servers},'' in \emph{USENIX Security
  Symposium (USENIX Security)}, 2018. [Online]. Available:
  \url{https://www.npmjs.com/package/safe-regex
  http://mp.binaervarianz.de/ReDoS_TR_Dec2017.pdf}
\BIBentrySTDinterwordspacing

\bibitem{barlas2022REDOS}
\BIBentryALTinterwordspacing
E.~Barlas, X.~Du, and J.~C. Davis, ``Exploiting input sanitization for regex
  denial of service,'' in \emph{Proceedings of the 44th International
  Conference on Software Engineering (ICSE '22)}, 2022, p. 883–895. [Online].
  Available: \url{https://doi.org/10.1145/3510003.3510047}
\BIBentrySTDinterwordspacing

\bibitem{Berglund2014REDOSTheory}
\BIBentryALTinterwordspacing
M.~Berglund, F.~Drewes, and B.~Van Der~Merwe, ``{Analyzing Catastrophic
  Backtracking Behavior in Practical Regular Expression Matching},''
  \emph{EPTCS: Automata and Formal Languages 2014}, vol. 151, pp. 109--123,
  2014. [Online]. Available: \url{https://arxiv.org/pdf/1405.5599.pdf}
\BIBentrySTDinterwordspacing

\bibitem{Weideman2016REDOSAmbiguity}
N.~Weideman, B.~van~der Merwe, M.~Berglund, and B.~Watson, ``{Analyzing
  matching time behavior of backtracking regular expression matchers by using
  ambiguity of NFA},'' in \emph{Lecture Notes in Computer Science (including
  subseries Lecture Notes in Artificial Intelligence and Lecture Notes in
  Bioinformatics)}, vol. 9705, 2016, pp. 322--334.

\bibitem{Weideman2017StaticExpressions}
N.~H. Weideman, ``{Static Analysis of Regular Expressions},'' \emph{MS Thesis},
  no. December, 2017.

\bibitem{Wustholz2017Rexploiter}
V.~Wustholz, O.~Olivo, M.~J.~H. Heule, and I.~Dillig, ``{Static Detection of
  DoS Vulnerabilities in Programs that use Regular Expressions},'' in
  \emph{International Conference on Tools and Algorithms for the Construction
  and Analysis of Systems (TACAS)}, 2017.

\bibitem{liu21sp}
\BIBentryALTinterwordspacing
Y.~Liu, M.~Zhang, and W.~Meng, ``Revealer: Detecting and exploiting regular
  expression denial-of-service vulnerabilities,'' in \emph{2021 2021 IEEE
  Symposium on Security and Privacy (SP)}, may 2021, pp. 1468--1484. [Online].
  Available:
  \url{https://doi.ieeecomputersociety.org/10.1109/SP40001.2021.00062}
\BIBentrySTDinterwordspacing

\bibitem{Kirrage2013rxxr}
J.~Kirrage, A.~Rathnayake, and H.~Thielecke, ``{Static Analysis for Regular
  Expression Denial-of-Service Attacks},'' \emph{Network and System Security},
  vol. 7873, pp. 35--148, 2013.

\bibitem{Rathnayake2014rxxr2}
A.~Rathnayake and H.~Thielecke, ``{Static Analysis for Regular Expression
  Exponential Runtime via Substructural Logics},'' \emph{CoRR}, 2014.

\bibitem{Sugiyama2014RegexLinearityAnalysis}
S.~Sugiyama and Y.~Minamide, ``{Checking Time Linearity of Regular Expression
  Matching Based on Backtracking},'' \emph{Information and Media Technologies},
  vol.~9, no.~3, pp. 222--232, 2014.

\bibitem{Sulzmann2017DerivAmbig}
\BIBentryALTinterwordspacing
M.~Sulzmann and K.~Z.~M. Lu, ``{Derivative-Based Diagnosis of Regular
  Expression Ambiguity},'' \emph{International Journal of Foundations of
  Computer Science}, vol.~28, no.~5, pp. 543--561, 4 2017. [Online]. Available:
  \url{http://arxiv.org/abs/1604.06644}
\BIBentrySTDinterwordspacing

\bibitem{Shen2018ReScueGeneticRegexChecker}
\BIBentryALTinterwordspacing
Y.~Shen, Y.~Jiang, C.~Xu, P.~Yu, X.~Ma, and J.~Lu, ``Rescue: Crafting regular
  expression dos attacks,'' in \emph{Proceedings of the 33rd ACM/IEEE
  International Conference on Automated Software Engineering (ASE 18)}, 2018,
  p. 225–235. [Online]. Available:
  \url{https://doi.org/10.1145/3238147.3238159}
\BIBentrySTDinterwordspacing

\bibitem{VanDerMerwe2017EvilRegexesHarmless}
\BIBentryALTinterwordspacing
B.~Van Der~Merwe, N.~Weideman, and M.~Berglund, ``{Turning Evil Regexes
  Harmless},'' in \emph{South African Institute of Computer Scientists and
  Information Technologists (SAICSIT)}, 2017. [Online]. Available:
  \url{https://doi.org/10.1145/3129416.3129440}
\BIBentrySTDinterwordspacing

\bibitem{codykenny17}
\BIBentryALTinterwordspacing
B.~Cody-Kenny, M.~Fenton, A.~Ronayne, E.~Considine, T.~McGuire, and M.~O'Neill,
  ``A search for improved performance in regular expressions,'' in
  \emph{Proceedings of the Genetic and Evolutionary Computation Conference
  (GECCO 17)}, 2017, p. 1280–1287. [Online]. Available:
  \url{https://doi.org/10.1145/3071178.3071196}
\BIBentrySTDinterwordspacing

\bibitem{li2020flashregex}
Y.~Li, Z.~Xu, J.~Cao, H.~Chen, T.~Ge, S.-C. Cheung, and H.~Zhao, ``Flashregex:
  deducing anti-redos regexes from examples,'' in \emph{2020 35th IEEE/ACM
  International Conference on Automated Software Engineering (ASE)}.\hskip 1em
  plus 0.5em minus 0.4em\relax IEEE, 2020, pp. 659--671.

\bibitem{claver2021regis}
M.~Claver, J.~Schmerge, J.~Garner, J.~Vossen, and J.~McClurg, ``Regis: Regular
  expression simplification via rewrite-guided synthesis,'' \emph{arXiv
  preprint arXiv:2104.12039}, 2021.

\bibitem{johnson2013don}
B.~Johnson, Y.~Song, E.~Murphy-Hill, and R.~Bowdidge, ``Why don't software
  developers use static analysis tools to find bugs?'' in \emph{2013 35th
  International Conference on Software Engineering (ICSE)}.\hskip 1em plus
  0.5em minus 0.4em\relax IEEE, 2013, pp. 672--681.

\bibitem{Davis2019LinguaFranca}
J.~C. Davis, L.~G. Michael~IV, C.~A. Coghlan, F.~Servant, and D.~Lee, ``{Why
  aren’t regular expressions a lingua franca? an empirical study on the
  re-use and portability of regular expressions},'' in \emph{The ACM Joint
  European Software Engineering Conference and Symposium on the Foundations of
  Software Engineering (ESEC/FSE)}, 2019.

\bibitem{hassan2022}
\BIBentryALTinterwordspacing
S.~A. Hassan, ``{Improving Developers' Understanding of Regex Denial of Service
  Tools through Anti-Patterns and Fix Strategies},'' Dec. 2022. [Online].
  Available: \url{https://zenodo.org/badge/latestdoi/575922405}
\BIBentrySTDinterwordspacing

\bibitem{kleene1951representation}
S.~Kleene, ``Representation of events in nerve nets and finite automata,'' RAND
  PROJECT AIR FORCE SANTA MONICA CA, Tech. Rep., 1951.

\bibitem{Sipser2006AutomataTheory}
M.~Sipser, \emph{Introduction to the Theory of Computation}.\hskip 1em plus
  0.5em minus 0.4em\relax Thomson Course Technology Boston, 2006, vol.~2.

\bibitem{Friedl2002MasteringRegexes}
J.~E. Friedl, \emph{{Mastering regular expressions}}.\hskip 1em plus 0.5em
  minus 0.4em\relax O'Reilly Media, Inc., 2002.

\bibitem{Davis2019RegexGeneralizability}
J.~C. Davis, D.~Moyer, A.~M. Kazerouni, and D.~Lee, ``{Testing Regex
  Generalizability And Its Implications: A Large-Scale Many-Language
  Measurement Study},'' in \emph{IEEE/ACM International Conference on Automated
  Software Engineering (ASE)}, 2019.

\bibitem{brabrand2010typed}
C.~Brabrand and J.~G. Thomsen, ``Typed and unambiguous pattern matching on
  strings using regular expressions,'' in \emph{Proceedings of the 12th
  international ACM SIGPLAN symposium on Principles and practice of declarative
  programming (PPDP 10)}, 2010, pp. 243--254.

\bibitem{rabin1959finite}
M.~O. Rabin and D.~Scott, ``Finite automata and their decision problems,''
  \emph{IBM journal of research and development}, vol.~3, no.~2, pp. 114--125,
  1959.

\bibitem{weber1991degree}
A.~Weber and H.~Seidl, ``On the degree of ambiguity of finite automata,''
  \emph{Theoretical Computer Science}, vol.~88, no.~2, pp. 325--349, 1991.

\bibitem{allauzen2008general}
C.~Allauzen, M.~Mohri, and A.~Rastogi, ``General algorithms for testing the
  ambiguity of finite automata,'' in \emph{International Conference on
  Developments in Language Theory}.\hskip 1em plus 0.5em minus 0.4em\relax
  Springer, 2008, pp. 108--120.

\bibitem{stearns1985equivalence}
R.~E. Stearns and H.~B. Hunt~III, ``On the equivalence and containment problems
  for unambiguous regular expressions, regular grammars and finite automata,''
  \emph{SIAM Journal on Computing}, vol.~14, no.~3, pp. 598--611, 1985.

\bibitem{davis2021using}
J.~C. Davis, F.~Servant, and D.~Lee, ``Using selective memoization to defeat
  regular expression denial of service (redos),'' in \emph{2021 IEEE Symposium
  on Security and Privacy (SP), Los Alamitos, CA, USA}, 2021, pp. 543--559.

\bibitem{Spencer1994RegexEngine}
H.~Spencer, ``{A regular-expression matcher},'' in \emph{Software solutions in
  C}, 1994, pp. 35--71.

\bibitem{Cox2007RegExBible}
\BIBentryALTinterwordspacing
R.~Cox, ``{Regular Expression Matching Can Be Simple And Fast (but is slow in
  Java, Perl, PHP, Python, Ruby, ...)},'' 2007. [Online]. Available:
  \url{https://swtch.com/~rsc/regexp/regexp1.html}
\BIBentrySTDinterwordspacing

\bibitem{Michael2019RegexesAreHard}
L.~G. Michael~IV, J.~Donohue, J.~C. Davis, D.~Lee, and F.~Servant, ``{Regexes
  are Hard : Decision-making, Difficulties, and Risks in Programming Regular
  Expressions},'' in \emph{IEEE/ACM International Conference on Automated
  Software Engineering (ASE)}, 2019.

\bibitem{davis2020impact}
J.~C. Davis, ``On the impact and defeat of regular expression denial of
  service,'' Ph.D. dissertation, Virginia Tech, 2020.

\bibitem{Crosby2003AlgorithmicComplexityAttacks}
S.~A. Crosby and D.~S. Wallach, ``{Denial of Service via Algorithmic Complexity
  Attacks},'' in \emph{USENIX Security}, 2003.

\bibitem{turonova22}
\BIBentryALTinterwordspacing
L.~Turo{\v n}ov{\'a}, L.~Hol{\'\i}k, I.~Homoliak, O.~Leng{\'a}l, M.~Veanes, and
  T.~Vojnar, ``Counting in regexes considered harmful: Exposing {ReDoS}
  vulnerability of nonbacktracking matchers,'' in \emph{31st USENIX Security
  Symposium (USENIX Security 22)}.\hskip 1em plus 0.5em minus 0.4em\relax
  Boston, MA: USENIX Association, Aug. 2022, pp. 4165--4182. [Online].
  Available:
  \url{https://www.usenix.org/conference/usenixsecurity22/presentation/turonova}
\BIBentrySTDinterwordspacing

\bibitem{bai2019exploring}
G.~R. Bai, B.~Clee, N.~Shrestha, C.~Chapman, C.~Wright, and K.~T. Stolee,
  ``Exploring tools and strategies used during regular expression composition
  tasks,'' in \emph{2019 IEEE/ACM 27th International Conference on Program
  Comprehension (ICPC)}.\hskip 1em plus 0.5em minus 0.4em\relax IEEE, 2019, pp.
  197--208.

\bibitem{wang2019exploring}
P.~Wang, G.~R. Bai, and K.~T. Stolee, ``Exploring regular expression
  evolution,'' in \emph{2019 IEEE 26th International Conference on Software
  Analysis, Evolution and Reengineering (SANER)}.\hskip 1em plus 0.5em minus
  0.4em\relax IEEE, 2019, pp. 502--513.

\bibitem{li2021redoshunter}
Y.~Li, Z.~Chen, J.~Cao, Z.~Xu, Q.~Peng, H.~Chen, L.~Chen, and S.-C. Cheung,
  ``$\{$ReDoSHunter$\}$: A combined static and dynamic approach for regular
  expression $\{$DoS$\}$ detection,'' in \emph{30th USENIX Security Symposium
  (USENIX Security 21)}, 2021, pp. 3847--3864.

\bibitem{mclaughlin2022regulator}
\BIBentryALTinterwordspacing
R.~McLaughlin, F.~Pagani, N.~Spahn, C.~Kruegel, and G.~Vigna, ``Regulator:
  Dynamic analysis to detect {ReDoS},'' in \emph{31st USENIX Security Symposium
  (USENIX Security 22)}, Boston, MA, Aug. 2022, pp. 4219--4235. [Online].
  Available:
  \url{https://www.usenix.org/conference/usenixsecurity22/presentation/mclaughlin}
\BIBentrySTDinterwordspacing

\bibitem{Petsios2017SlowFuzz}
\BIBentryALTinterwordspacing
T.~Petsios, J.~Zhao, A.~D. Keromytis, and S.~Jana, ``{SlowFuzz: Automated
  Domain-Independent Detection of Algorithmic Complexity Vulnerabilities},'' in
  \emph{Computer and Communications Security (CCS)}, 2017. [Online]. Available:
  \url{https://arxiv.org/pdf/1708.08437.pdf}
\BIBentrySTDinterwordspacing

\bibitem{wei2018singularity}
J.~Wei, J.~Chen, Y.~Feng, K.~Ferles, and I.~Dillig, ``Singularity: Pattern
  fuzzing for worst case complexity,'' in \emph{Proceedings of the 2018 26th
  ACM Joint Meeting on European Software Engineering Conference and Symposium
  on the Foundations of Software Engineering (ESEC/FSE 18)}, 2018, pp.
  213--223.

\bibitem{noller2018badger}
Y.~Noller, R.~Kersten, and C.~S. P{\u{a}}s{\u{a}}reanu, ``Badger: complexity
  analysis with fuzzing and symbolic execution,'' in \emph{Proceedings of the
  27th ACM SIGSOFT International Symposium on Software Testing and Analysis
  (ISSTA 18)}, 2018, pp. 322--332.

\bibitem{meng2018rampart}
W.~Meng, C.~Qian, S.~Hao, K.~Borgolte, G.~Vigna, C.~Kruegel, and W.~Lee,
  ``Rampart: Protecting web applications from cpu-exhaustion denial-of-service
  attacks,'' in \emph{27th USENIX Security Symposium (USENIX Security 18)},
  2018, pp. 393--410.

\bibitem{blair2020hotfuzz}
W.~Blair, A.~Mambretti, S.~Arshad, M.~Weissbacher, W.~Robertson, E.~Kirda, and
  M.~Egele, ``Hotfuzz: Discovering algorithmic denial-of-service
  vulnerabilities through guided micro-fuzzing,'' \emph{arXiv preprint
  arXiv:2002.03416}, 2020.

\bibitem{chida2022repairing}
N.~Chida and T.~Terauchi, ``Repairing dos vulnerability of real-world
  regexes,'' in \emph{2022 IEEE Symposium on Security and Privacy (SP)}.\hskip
  1em plus 0.5em minus 0.4em\relax IEEE, 2022, pp. 2060--2077.

\bibitem{bai2021runtime}
Z.~Bai, K.~Wang, H.~Zhu, Y.~Cao, and X.~Jin, ``Runtime recovery of web
  applications under zero-day redos attacks,'' in \emph{2021 IEEE Symposium on
  Security and Privacy (SP)}.\hskip 1em plus 0.5em minus 0.4em\relax IEEE,
  2021, pp. 1575--1588.

\bibitem{atre2022surgeprotector}
N.~Atre, H.~Sadok, E.~Chiang, W.~Wang, and J.~Sherry, ``Surgeprotector:
  Mitigating temporal algorithmic complexity attacks using adversarial
  scheduling,'' in \emph{Proceedings of the 2022 Conference of the ACM Special
  Interest Group on Data Communication (SIGCOMM), New York, NY, USA}, 2022.

\bibitem{davis2018sense}
J.~C. Davis, E.~R. Williamson, and D.~Lee, ``A sense of time for javascript and
  node. js: First-class timeouts as a cure for event handler poisoning,'' in
  \emph{27th USENIX Security Symposium (USENIX Security 18)}, 2018, pp.
  343--359.

\bibitem{demoulin2019detecting}
H.~M. Demoulin, I.~Pedisich, N.~Vasilakis, V.~Liu, B.~T. Loo, and L.~T.~X.
  Phan, ``Detecting asymmetric application-layer denial-of-service attacks
  in-flight with finelame,'' in \emph{2019 USENIX Annual Technical Conference
  (USENIX ATC)}, 2019, pp. 693--708.

\bibitem{brzozowski1964derivatives}
J.~A. Brzozowski, ``Derivatives of regular expressions,'' \emph{Journal of the
  ACM (JACM)}, vol.~11, no.~4, pp. 481--494, 1964.

\bibitem{Thompson1968LinearRegexAlgorithm}
K.~Thompson, ``{Regular Expression Search Algorithm},'' \emph{Communications of
  the ACM (CACM)}, 1968.

\bibitem{Cox2010RE2Implementation}
\BIBentryALTinterwordspacing
R.~Cox, ``{Regular Expression Matching in the Wild},'' 2010. [Online].
  Available: \url{https://swtch.com/~rsc/regexp/regexp3.html}
\BIBentrySTDinterwordspacing

\bibitem{RustRegexDocs}
T.~R.~P. Developers, ``regex - rust,''
  \url{https://docs.rs/regex/1.1.0/regex/}.

\bibitem{GoRegexDocs}
Google, ``regexp - go,'' \url{https://golang.org/pkg/regexp/}.

\bibitem{saarikivi2019symbolic}
O.~Saarikivi, M.~Veanes, T.~Wan, and E.~Xu, ``Symbolic regex matcher,'' in
  \emph{International Conference on Tools and Algorithms for the Construction
  and Analysis of Systems}.\hskip 1em plus 0.5em minus 0.4em\relax Springer,
  2019, pp. 372--378.

\bibitem{holik2019succinct}
L.~Hol{\'\i}k, O.~Leng{\'a}l, O.~Saarikivi, L.~Turo{\v{n}}ov{\'a}, M.~Veanes,
  and T.~Vojnar, ``Succinct determinisation of counting automata via sphere
  construction,'' in \emph{Asian Symposium on Programming Languages and
  Systems}.\hskip 1em plus 0.5em minus 0.4em\relax Springer, 2019, pp.
  468--489.

\bibitem{turovnova2020regex}
\BIBentryALTinterwordspacing
L.~Turo\v{n}ov\'{a}, L.~Hol\'{\i}k, O.~Leng\'{a}l, O.~Saarikivi, M.~Veanes, and
  T.~Vojnar, ``Regex matching with counting-set automata,'' in
  \emph{Object-oriented Programming, Systems, Languages, and Applications
  (OOPSLA 20)}, Virtual, November 2020. [Online]. Available:
  \url{https://doi.org/10.1145/3428286}
\BIBentrySTDinterwordspacing

\bibitem{sung22}
S.~Sung, H.~Cheon, and Y.-S. Han, ``How to settle the redos problem: Back
  to the classical automata theory,'' in \emph{Implementation and Application
  of Automata}, P.~Caron and L.~Mignot, Eds.\hskip 1em plus 0.5em minus
  0.4em\relax Cham: Springer International Publishing, 2022, pp. 34--49.

\bibitem{moller2010dk}
A.~M{\o}ller, ``dk. brics. automaton--finite-state automata and regular
  expressions for java,'' 2010.

\bibitem{aho2020compilers}
A.~V. Aho, M.~S. Lam, R.~Sethi, and J.~D. Ullman, \emph{Compilers: principles,
  techniques and tools}, 2020.

\bibitem{safeRegexHomePage}
substack and Davis, ``safe-regex,''
  \url{https://www.npmjs.com/package/safe-regex}, 2013.

\bibitem{ANTLRPCRE}
``antlr-pcre,''
  \url{https://web.archive.org/web/20210826063830/https://github.com/bkiers/pcre-parser}.

\bibitem{ting10}
\BIBentryALTinterwordspacing
K.~M. Ting, \emph{Precision and Recall}.\hskip 1em plus 0.5em minus 0.4em\relax
  Boston, MA: Springer US, 2010, pp. 781--781. [Online]. Available:
  \url{https://doi.org/10.1007/978-0-387-30164-8_652}
\BIBentrySTDinterwordspacing

\bibitem{chida2020automatic}
N.~Chida and T.~Terauchi, ``Automatic repair of vulnerable regular
  expressions,'' \emph{arXiv preprint arXiv:2010.12450}, 2020.

\bibitem{wilcoxon1992individual}
F.~Wilcoxon, ``Individual comparisons by ranking methods,'' in
  \emph{Breakthroughs in statistics}.\hskip 1em plus 0.5em minus 0.4em\relax
  Springer, 1992, pp. 196--202.

\bibitem{ellis2010essential}
P.~D. Ellis, \emph{The essential guide to effect sizes: Statistical power,
  meta-analysis, and the interpretation of research results}.\hskip 1em plus
  0.5em minus 0.4em\relax Cambridge university press, 2010.

\bibitem{pypiHomePage}
``Pypi -- the python package index,'' \url{https://pypi.python.org/pypi}.

\bibitem{NPMHomePage}
``npm,'' \url{https://www.npmjs.com}.

\bibitem{kitchenham2008personal}
B.~A. Kitchenham and S.~L. Pfleeger, ``Personal opinion surveys,'' in
  \emph{Guide to advanced empirical software engineering}.\hskip 1em plus 0.5em
  minus 0.4em\relax Springer, 2008, pp. 63--92.

\bibitem{siegmund2014measuring}
J.~Siegmund, C.~K{\"a}stner, J.~Liebig, S.~Apel, and S.~Hanenberg, ``Measuring
  and modeling programming experience,'' \emph{Empirical Software Engineering},
  vol.~19, no.~5, pp. 1299--1334, 2014.

\bibitem{cha2020principled}
A.~Cha, E.~Wittern, G.~Baudart, J.~C. Davis, L.~Mandel, and J.~A. Laredo, ``A
  principled approach to graphql query cost analysis,'' in \emph{Proceedings of
  the 28th ACM Joint Meeting on European Software Engineering Conference and
  Symposium on the Foundations of Software Engineering (ESEC/FSE 20)}, 2020,
  pp. 257--268.

\end{thebibliography}

\end{document}